\documentclass[%
nofootinbib,
 amsmath,amssymb,
 aps,
]{revtex4-2}

\usepackage[switch]{lineno} 
\usepackage{amsfonts}
\usepackage{graphicx}
\usepackage{epstopdf}
\usepackage{cancel}
\usepackage{graphicx}
\usepackage[dvipsnames]{xcolor}
\usepackage[section]{placeins}
\usepackage{subfig}
\usepackage{epsfig}
\usepackage{amsmath}
\usepackage{amsthm}
\usepackage{mathbbol}
\usepackage[ruled,lined]{algorithm2e}
\usepackage{xfrac}
\usepackage{comment}
\definecolor{lightblue}{rgb}{0.49, 0.49, 1}
\definecolor{indigo}{rgb}{0.29, 0.0, 0.51}
\definecolor{magenta}{rgb}{255,0,255}

\usepackage{booktabs}
\usepackage{hyperref}
\usepackage{cleveref}
\usepackage{comment}

\crefformat{figure}{#2Fig.~#1#3}
\crefmultiformat{figure}{#2Figs.~#1#3}{ and~#2#1#3}{, #2#1#3}{ and~#2#1#3}

\crefformat{theorem}{#2Theorem~#1#3}

\crefformat{enumi}{#2Step~#1#3}
\crefmultiformat{enumi}{#2Steps~#1#3}{ and~#2#1#3}{, #2#1#3}{ and~#2#1#3}

\crefrangeformat{figure}{#3Figs.~#1#4 to~#5#2#6}

\crefformat{appendix}{#2Appendix~#1#3}

\crefformat{section}{#2\S~#1#3}
\crefmultiformat{section}{#2\S\S~#1#3}{ and~#2#1#3}{, #2#1#3}{ and~#2#1#3}

\crefformat{algorithm}{#2Algorithm~#1#3}
\crefmultiformat{algorithm}{#2Algorithms~#1#3}{ and~#2#1#3}{, #2#1#3}{ and~#2#1#3}

\crefformat{equation}{#2Eq.~(#1)#3}
\crefmultiformat{equation}{#2Eqs.~(#1)#3}{ and~#2(#1)#3}{, #2(#1)#3}{ and~#2(#1)#3}

\crefformat{table}{#2Table~#1#3}
\crefmultiformat{table}{#2Tables~#1#3}{ and~#2#1#3}{, #2#1#3}{ and~#2#1#3}

\usepackage{verbatim}
\usepackage{chngcntr}
\usepackage{comment}

\newtheorem{theorem}{Theorem}

\begin{document}


\title{Flow through Pore-Size Graded Membrane Pore Network}

\thanks{This work was supported by NSF Grants No. DMS-1615719, DMS-2133255 and DMS-2201627.}

\author{Binan Gu}
\email{Corresponding emails: bg263@njit.edu, kondic@njit.edu and linda.cummings@njit.edu}
\author{Lou Kondic}
\author{Linda J. Cummings}

\date{\today}

\begin{abstract}
Pore-size gradients are often used in the design of membrane filters to increase filter lifetime and ensure fuller use of the initial membrane pore volume. In this work, we impose pore-size gradients in the setting of a membrane filter with an internal network of interconnected tube-like pores. We model the flow and foulant transport through the filter using the Hagen-Poiseuille framework coupled with advection equations via conservation of fluid and particle flux, with adsorption as the sole fouling mechanism. We study the influence of pore-size gradient on performance measures such as total filtrate throughput and accumulated contaminant concentration at the membrane downstream pore outlets. Within the limitations of our modeling assumptions we find that there is an optimal pore-radius gradient that maximizes filter efficiency independent of maximum pore length (an input parameter that influences the structure of the pore network), and that filters with longer characteristic pore length perform better. 

\end{abstract}

\maketitle


\section{Introduction}
Membrane filtration is an industrial process that uses porous material to separate contaminants from a feed solution. It is crucial to commercial processes such as waste water treatment~\cite{omni2012}, radioactive sludge removal~\cite{chiera2021}, beer clarification~\cite{Lipnizki2015} and membrane bioreactors~\cite{Dizge2011}, among many others. Filtration also underpins many daily household appliances including water purifiers~\cite{barnaby2017}, air filters~\cite{Yu2009,vijayan2015,Sublett2011,LIU2017375}, and grease filters~\cite{Wong2007}. To design an ideal filter, one aims to tailor the geometric features of the filter (specifically, the pores' size, shape and connectivity) so that impurities are removed efficiently, while producing a required amount of filtrate up to a certain standard of purity.

Membrane filtration employs a wide variety of pressure-driven separation methods, distinguished by the scales of pore sizes at which they operate. For example, microfiltration is effective in sieving solids and bacteria;  ultrafiltration is often employed in virus and toxin removal; nanofiltration is a popular final step for water treatment that removes major monovalent ions such as chloride and sodium; and reverse osmosis separates multivalent ions that escape from the previous methods by applying mechanical pressure to overcome osmotic pressure~\cite{bruggen2018}. A wide range of materials may be used in membrane manufacture, but membrane materials in common use are roughly divided into two categories: polymeric and ceramic~\cite{Goswami2020}. Most polymeric membranes 
are made with low-cost organic materials and are popular in industrial applications; however, they tolerate large thermal fluctuations or harsh chemical environments poorly. The more expensive ceramic membranes overcome these drawbacks via their chemical composition ({\it e.g.,} metal oxides), while producing higher fluxes due to their greater hydrophilicity. Designing layered or composite filters that incorporate the merits and disadvantages of both materials has become an active area of research~\cite{cera_poly_composite}. Furthermore, membrane filters are manufactured with many different spatial configurations of membrane materials (leading to differently-structured pores) such as node-fibril~\cite{RanjbarzadehDibazar2017}, flat-sheet~\cite{SANO2020123015} and multitube~\cite{Ma2018}, either mimicking natural filters found in plants and animal organs, or resulting from careful design considerations.

It is clear from the studies cited above (and many others too numerous to cite) that membrane filter design, using a combination of materials and pore layouts, has profound implications for filter performance. A well-designed membrane filter not only maintains particle retention capability but also provides sufficient output (filtrate) to serve the immediate needs of the application. Common membrane designs incorporate structural variations at the microscale (connectivity of interior pores, pore branching, etc.) as well as the macroscale; {\it e.g.} pleated filters (also known as spiral-wound modules~\cite{Basile2010}) and multilayered membrane filters~\cite{multi_onur_2017}. This last class of membranes has garnered particular attention from industrialists and practitioners for their versatility in applications. In multilayered membrane filters, single-layer membrane filters, each with a different characteristic pore size, may be laminated to form a composite membrane with a distribution of pore sizes in its depth (in addition to any intra-layer pore-size variations, which are assumed to be less important; but see Gu~{\it et al.}~\cite{gu_jms_2022} for more discussion of this feature). Common practice is to place the layers with larger pores upstream, allowing more fluid to pass through and providing more surface area to capture particles; and layers with smaller pores downstream to capture any particles that may escape from upstream layers.

We emphasize that the notions of porosity gradient and pore-size gradient are inherently different, though they are sometimes related; and the terms are often used interchangeably in various contexts within current literature. To illustrate, consider first a simple membrane structure where the membrane upstream and downstream surfaces are connected by single continuous pores (a ``track-etched'' type structure~\cite{Apel2001}). Suppose such pores have circular cross-section (but with depth-varying radius) and a straight axis perpendicular to the membrane, then if one identifies local pore radius with pore size, a direct relation may be made between pore-size gradient and porosity gradient: a pore-size (radius) gradient does induce a porosity gradient of the same sign (or vice versa). However, for a more general network of pores, one can easily design a network that has decreasing pore size in the depth but no porosity gradient (or even a porosity gradient of opposite sign), by appropriately increasing the number density of pores with depth. With this distinction between porosity and pore-size gradient in mind, we briefly review current relevant literature that considers either type of gradient (often both are present), noting the key findings and motivating the work of the current paper.

Many experimental studies have shown that improved performance can be achieved with {\bf porosity-graded} multilayered membrane filters. For example, in terms of membrane performance, the ability of porosity-graded ceramic filters to delay internal fouling when compared to common homogeneous ceramic filters has been discussed~\cite{Ng2020}; and the improved throughput production and foulant removal capability of a photocatalytic membrane with hierarchical porosity was investigated~\cite{Goei2013}. Progress has also been made on the manufacturing side. Improved tunability of the physical membrane characteristics in porosity-graded membrane filter assembly has been demonstrated by Amin {\it et al.}~\cite{Amin2017}; a novel fabrication strategy of porosity-graded porous foams via 3D printing has been discussed by Cappaso~{\it et al.}~\cite{capasso2020}; and recent advances in additive manufacturing techniques for porous materials with controllable structure have been reviewed by Guddati {\it et al.}~\cite{guddati2019}. On the other hand, much attention has also been given to manufacturing {\bf pore-size-graded} filters with desired characteristics, such as the work of Dong {\it et al.}~\cite{Dong2022} on fabricating air filters represented as pore-size-graded networks (mimicking bryophyte leafs), and of Harley~{\it et al.}, who present a novel strategy to build porous tubular scaffolds with prescribed pore-size gradient~\cite{HARLEY2006866}. Kosiol~{\it et al.} use gold nanoparticles as a probe to estimate the pore-size gradient in commercial and non-commercial parvovirus retentive membranes, as these gradient values were found to correlate strongly with virus retention~\cite{Kosiol2018}. The advantages of pore-size-graded filters in applications such as tissue engineering~\cite{DiLuca2016,MANNELLA201531,SOBRAL20111009} and fuel cells~\cite{OH2015186} have also been discussed.

Several theoretical groups have also contributed to the breadth of the study on multilayered membrane filters with porosity and/or pore size gradients via mathematical modeling and numerical studies; for example studying filter performance optimization as a function of pore-size gradient within a simply-structured membrane~\cite{ian2021,Griffiths2016,yixuan2020,sanaei2019}; carrying out numerical simulation and analysis of performance of simple multilayered filter structures~\cite{fong_multi}, and investigating geometric and topological properties of membrane networks~\cite{gu_network_2021}. There is also significant work on techniques to probe the microstructure of membrane filters, including imaging techniques used to recover network representations of membranes structures in 3D~\cite{imperial,3d_image} that motivate and inform our modeling work using pore networks.  

To further motivate the current work, we summarize our previous efforts on the modeling of membrane filtration, all of which assumed that membrane pores are circularly-cylindrical connected tubes. Ref.~\cite{gu2020} studied layered membrane filters with three different (very simple) internal pore structures, having varying degrees of connectivity. A pore-size gradient was introduced using a geometric parameter that prescribes the initial pore radius in each layer. In subsequent work \cite{gu_network_2021}, the authors generalized those simple pore configurations to membranes with a random network of pores and formulated the mathematical equations for flow, transport and fouling on such networks, on which the current work is based. Lastly, pore-size (radius) variations were imposed on these membrane pore networks, and their effect studied~\cite{gu_jms_2022}. Initial pore radii follow a uniform distribution, centered about a fixed average, independent of the pore's depth in the membrane. However, pore-size gradient was not considered in the last two works since the focus was to draw out basic geometric factors of the network structure that influence membrane filter performance. 

Building on this previous work, the current paper focuses on network models of {\bf pore-size-graded} filters with constant porosity across the filter. We specifically exclude porosity variations so that the impact of pore-size gradient alone can be elucidated; and also because, even in homogeneous membranes, porosity has been shown to influence filtration performance rather strongly, to the extent that no benefit would be anticipated by having porosity decrease in the membrane depth (see~\cite{gu_network_2021,gu_jms_2022}, for example). Our goals are to model a membrane filter with a pore-size gradient and then to probe and explain the influence of this gradient on membrane performance metrics such as total filtrate throughput and particle retention (only adsorptive particle fouling is considered in this work). For the first goal, we model the membrane filter as a network of circularly-cylindrical pores. We introduce the pore-size gradient by dividing the membrane into bands of equal thickness (within each of which initial pore radius is constant), and designating a linearly decreasing sequence of radius values for pores from upstream to downstream bands. We generate membrane pore networks with such a banded structure following a random network generation procedure, adapted from that proposed by Gu~{\it et al.}~\cite{gu_network_2021}. We further impose that the porosity of each band, an influential geometric feature of membrane pore networks, is approximately equal across all bands, so that we reveal the sole influence of pore size (radius) gradient on membrane filter performance. In addition to studying pore-size gradient variations, we also consider variations in maximum pore length, another model input parameter that controls the geometric structure of the pore network. 

The paper is outlined as follows: in \cref{sec:modelling}, we describe the details of the mathematical model, first introducing how the banded pore networks are constructed with specified pore-size gradient in \cref{sec:graded} and then presenting the governing equations for fluid flow and foulant particle transport in \cref{sec:governing}; in \cref{sec:perf_metrics} we define the performance metrics we use to compare our membrane pore networks; in \cref{sec:algorithm}, we streamline the pore-size-graded network generation procedures into an algorithm; in \cref{sec:grad_scales}, we provide the appropriate physical scales of the problem and then summarize the model in nondimensional form; in \cref{sec:grad_results}, we present and explain our observations; and in \cref{sec:conclusion_grad}, we conclude our findings.

\section{Mathematical Modeling}\label{sec:modelling}
In this section, we introduce a mathematical model that captures the multilayered membrane structure using a pore network representation. We first describe the general network generation protocol and how this creates pore junctions and cylindrical pores, and define our computational domain. After introducing our notions of pore-size gradient, band radius and band porosity, we provide details of the methodology by which we generate radius-graded banded networks under specific physical constraints. Lastly, we briefly present the governing equations for flow of a feed solution through the membrane filter, for transport of the foulant particles carried by the feed, and for the pore-radius evolution under fouling, as well as the solution techniques for these equations when they are posed on a network of interconnected pores.

\subsection{Pore Size-Graded Networks}\label{sec:graded}
We model a representative unit of a membrane filter as a block of porous material that occupies a cube with side length $W$ (see \cite{gu_network_2021} for a similar setup) and contains a network of pores. Each pore is assumed to be circularly cylindrical and thus fully characterized by its length and radius. We use the terms ``pore size'' and ``pore radius'' interchangeably from hereon. The unit consists of a membrane top surface with pore inlets, interior pore junctions (vertices of the network), pores (edges of the network), and a bottom membrane surface with pore outlets. The membrane unit is generated as follows: interior junctions are points represented by Euclidean coordinates in $\mathbb{R}^3$, uniformly randomly placed in a rectangular box with height $2W$ and square cross section of side length $W$. Pores are constructed as slender circular cylinders, with axes along straight lines that connect the junctions according to a periodic connection metric (see \cref{app:vertex_edge_defn}). More specifically, we connect junctions (possibly through the side boundaries) when they lie within a distance of $A_{\rm max}$ and at least $A_{\rm min}$ away from each other. These two parameters are referred to as the maximum and minimum pore lengths, respectively. Inlets and outlets are the intersection points between the pores thus generated and two horizontal planes at heights $Z=0.5W$ and $Z=1.5W$ respectively (thus $Z$ is the coordinate perpendicular to the membrane surfaces). The space created between these planes forms our computational domain -- a cube with side length $W$ (referred to as {\bf the domain} in the following), while the parts exterior to this cube are discarded.

In the following sections, we first introduce the notion of pore-size (radius) gradients, as studied in this paper, and define band porosity. In addition, in order to tune membrane porosity, we devise a rule to determine the number of pore junctions to be placed (randomly) in the domain to complete the banded network generation protocol.

\subsubsection{Bands and Radius Gradient}\label{subsec:bands}

We introduce a pore-radius gradient by first dividing the domain in the $Z$-direction (the coordinate direction perpendicular to the membrane upstream and downstream surfaces) into $m$ bands, each of thickness $W/m$. \cref{fig:3d_banded_network} shows a 2D representation of a 3D schematic. Let $\mathcal{V}_k$ and $\mathcal{E}_k$ be the set of junctions and pores in the $k^{\rm th}$ band (numbered from the upstream surface; see the detailed definitions in \cref{eq:banded_vertex_set,eq:edge_set} in \cref{app:vertex_edge_defn}). Each pore in $\mathcal{E}_k$ is assigned an initial radius
\begin{equation}
    R_k = R_m + (m-k)sW, \quad 1\leq k \leq m,
    \label{eq:band_radius}
\end{equation}
where $s\geq 0$ is the radius gradient and $R_m$ is the radius of pores in the bottom band ($k=m$). This consideration assumes that the initial pore radius in each band is a constant, and that pore size always decreases in the depth of the membrane. We say that a pore belongs to the $k^{\rm th}$ band when the largest proportion of its total length lies inside the $k^{\rm th}$ band, and we then assign $R_k$ as its initial radius (see the different thicknesses and color coding of pores across the bands in the schematic of \cref{fig:3d_banded_network}). We refer to $R_k$ as the $k^{\rm th}$ {\bf band radius} from hereon. We also call networks with $s=0$ {\bf uniform networks} and those with nonzero $s$ values {\bf graded networks}.

\begin{figure}
    \centering
    \includegraphics[scale=0.5]{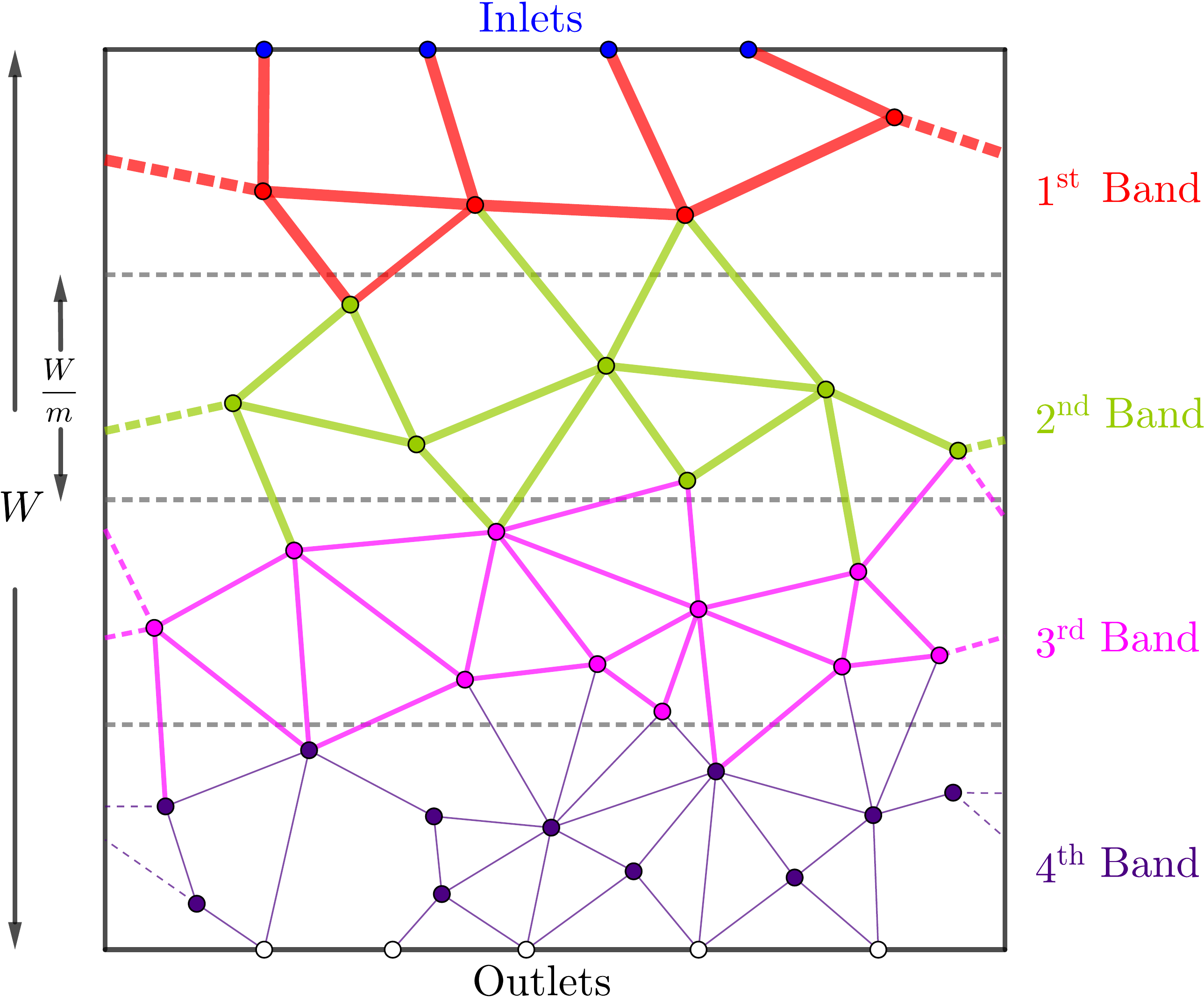}
    \caption{Schematic of a 3D banded network represented in 2D. Colored junctions and pores correspond to each band as follows: {\color{red}red $1^{\rm st}$ band}; {\color{LimeGreen}green $2^{\rm nd}$ band}; {\color{magenta}magenta $3^{\rm rd}$ band} and {\color{indigo}indigo $4^{\rm th}$ band}. {\color{blue}Blue dots} are inlets. White dots are outlets. Dashed lines are pores created by the periodic boundary conditions (see \cref{eq:connection_metric} in \cref{app:vertex_edge_defn}).
    }
    \label{fig:3d_banded_network}
\end{figure}

In this work, to reduce the number of degrees of freedom and to make our comparisons of different composite membranes as ``fair'' as possible, we impose the constraint that the average pore radius across all bands is equal to some value $R_0$, for any $R_m$ and $s$ values. More precisely, for each graded network, we find a range of $\left(R_m,s\right)$ pairs corresponding to an average pore radius $R_0$ across the bands, and compare the performance of these networks versus their uniform counterparts with radius $R_0$ ($s=0$). The average radius $R_0$ across the $m$ bands satisfies (using \cref{eq:band_radius})
\begin{equation}
\left(\textit{Constraint 1a}\right)\qquad R_0 = \frac{1}{m}\sum_{k=1}^m R_k = R_{m}+\frac{sW}{2}\left(m-1\right).
\label{eq:constraint_1a}
\end{equation}
Thus, once $R_0$ is prescribed, we can choose $s$ and then uniquely determine $R_m$.

We here point out that to generate a banded network, in addition to declaring band radii using \cref{eq:band_radius}, we must also specify an initial number of randomly-placed junctions in each band. This in general leads to different band porosities (void volume in the band divided by volume of the containing slab, discussed in detail later). As the goal of our work is to isolate the influence of pore-size (radius) gradients, we choose to enforce that each band has the same porosity. To achieve this, we devise a guess-and-correct procedure detailed in the following sections.

\subsubsection{Band and Membrane Porosity}\label{sec:band_and_membrane_porosity}
Before we introduce the banded network generation procedures, we clarify the crucial definition of band porosity. The {\bf band porosity} of the $k^{\rm th}$ band of a graded network with $m$ bands is given by
\begin{equation}
\Phi_k = \frac{\frac{\pi}{2} \sum_{e_{ij}\in \mathcal{E}} R_{ij}^2 L_{k,ij}}{W^3/m},
\label{eq:band_porosity}
\end{equation}
where $e_{ij}$ is the pore connecting junctions $i$ and $j$ with radius $R_{ij}$, and $L_{k,ij}$ measures the length of $e_{ij}$ that lies within the $k^{\rm th}$ band (by which definition $L_{k,ij}=0$ if $e_{ij}$ has no component within layer $k$; see \cref{app:vertex_edge_defn} for a rigorous definition of $L_{k,ij}$), and $\mathcal{E}$ is the set of all edges (per \cref{eq:edge}). The sum in \cref{eq:band_porosity} is over all pores since $L_{k,ij}$ includes the contributions from pores that cross multiple bands. The numerator is the void volume of the $k^{\rm th}$ band (volume of empty space), while the denominator is the volume of the rectangular slab with a square cross section of side length $W$ and height $W/m$.

The {\bf membrane porosity} is given by
\begin{equation}
\Phi = \frac{1}{m}\sum_{k=1}^m \Phi_k,
\label{eq:membrane_porosity}
\end{equation}
that is, the average band porosity across the bands. This definition is equivalent to the sum of total void volume divided by the volume of the cube with side length $W$. 

In general, the band porosities and overall membrane porosity are functions of time, since pore radii evolve due to foulant particle deposition and adsorption; \cref{eq:band_porosity,eq:membrane_porosity} hold pointwise in time. In the following \cref{sec:const_radius}, however, we are concerned only with describing the initial membrane structure, hence we adopt a simpler notation by dropping the time dependence as we work only with {\bf initial} porosities. 

\subsubsection{Constant Radius in Each Band}\label{sec:const_radius}
Here we present the methodology of generating radius-graded membrane networks. In this section only, for notational simplicity, $\Phi_k$ and $\Phi$ denote the {\bf initial} porosity in the $k^{\rm th}$ band and {\bf initial} membrane porosity, respectively. We assign an initial porosity value $\Phi$ for the membrane pore structure, which, in all simulations presented here, is set as $\Phi = 0.6$, a value typical for commercial filters. The principal aim of this section is to estimate the number of points (pore junctions) needed in each band to generate networks that satisfy prescribed constraints, including \cref{eq:constraint_1a} and more to be detailed below.

To isolate the effect of pore-radius gradients on membrane performance, we enforce  that every band has approximately the same initial porosity (we cannot insist that band porosities be exactly equal due to the random nature of the network generation protocol). This (soft) constraint is given by
\begin{equation}
\left(\textit{Constraint 2a}\right)\qquad \Phi_n \approx \Phi_k,\quad n,k=1,\ldots,m,
\label{eq:constraint_2a}
\end{equation}
where the ``$\approx$'' is to be made precise in the algorithm introduced below. Since the initial pore radius is constant within each band, \cref{eq:band_porosity} simplifies to
\begin{equation}
\Phi_k := \frac{\frac{\pi}{2} R_k^2 \sum_{e_{ij}\in \mathcal{E}} L_{k,ij}}{W^3/m}.
\label{eq:band_porosity_const}
\end{equation}
Using \cref{eq:constraint_2a,eq:band_porosity_const}, for arbitrary bands $n$ and $k$, we have
\begin{equation}
\frac{W^3}{m}\Phi_n = \frac{\pi}{2}R_{n}^{2} \sum_{e_{ij}\in \mathcal{E}} L_{n,ij} \approx \frac{\pi}{2}R_k^{2}\sum_{e_{ij}\in \mathcal{E}}L_{k,ij} = \frac{W^3}{m}\Phi_k, \qquad \implies \qquad R_{n}^{2} \sum_{e_{ij}\in \mathcal{E}} L_{n,ij} \approx R_k^{2}\sum_{e_{ij}\in \mathcal{E}}L_{k,ij}.
\label{eq:local_porosity_approx}
\end{equation}

We now relate $N_k$, the number of pore junctions randomly placed in the $k^{\rm th}$ band, to the sum in \cref{eq:local_porosity_approx} and use the result to provide an estimate for $N_k$, to generate pore networks that satisfy the constraints given in \cref{eq:constraint_1a,eq:constraint_2a}. We use basic probabilistic arguments to deduce that $\sum_{e_{ij}\in \mathcal{E}} L_{k,ij}$, total edge length in the $k^{\rm th}$ band, scales with $N_k\left(N_k-1\right)$ and other terms common to all other bands (so that they cancel from each side of \cref{eq:local_porosity_approx}). Details of this derivation are in \cref{app:band_estimate}. \cref{eq:local_porosity_approx} then reduces to
\begin{equation}
R_n^{2}N_n\left(N_n-1\right) \approx R_k^{2}N_k\left(N_k-1\right),\quad \forall n,k=1,\ldots,m,\quad n\neq k.
\label{eq:num_k_relationship_1}
\end{equation}
Since we consider only situations where pore size decreases in the membrane depth, a nonzero pore-size gradient $s>0$ implies that $R_m$, the radius in the $m^{\rm th}$ band, is the smallest (per~\cref{eq:band_radius}). Thus, the $m^{\rm th}$ band requires more points than other bands to satisfy \cref{eq:num_k_relationship_1}, and motivates our initializing our algorithm with this band (it is more computationally efficient, since this choice minimizes the errors incurred in the sequential process outlined next). To estimate the $N_k$ sequentially, we first prescribe $s$, the radius gradient. This fixes each band radius $R_k$ via {\bf Constraint 1}, \cref{eq:constraint_1a}. We then make several guesses for $N_m$ and stop when we find a value such that $\Phi_m\approx \Phi$. With $N_m$ determined, we employ \cref{eq:band_radius,eq:num_k_relationship_1} to obtain the relationship
\begin{equation}
R_m^{2}N_m\left(N_m-1\right) \approx \left(R_m + \left(m-k\right)s\right)^{2}N_k\left(N_k-1\right),
\label{eq:num_k_relationship_2}
\end{equation}
which we solve to estimate the number of junctions, $N_k$, for all other bands\footnote{We recall that the initial domain is a rectangular prism of height $2W$ with square cross-sections of side length $W$. The regions, $0\leq Z<0.5W$ and $1.5W < Z \leq 2W$, outside the central cube (see~\cref{fig:3d_banded_network}) before cutting, have the same point density as the $1^{\rm st}$ band and the $m^{\rm th}$ band, respectively.}. 

The relationship \cref{eq:num_k_relationship_2} is by no means exact, and can fail by some margin to guarantee equal porosity for each band when the gradient $s$ becomes too large, but it provides a useful starting point. After estimating the $N_k$ values as described, we compute the corresponding $\Phi_k$ and check their proximity to the prescribed value $\Phi$. We correct $\Phi_k$ to $\Phi$ by adding or removing nodes randomly. This correction procedure starts with the $m^{\rm th}$ band and proceeds upstream. We iterate this procedure until the network achieves band porosities close to $\Phi$ within a prescribed relative tolerance $\epsilon=0.005$. Variations in porosity as such are sufficiently small that they have a negligible effect on our results for performance metrics (see~\cite{gu_jms_2022} for details on the effect of porosity variations).

\subsection{Governing Equations}\label{sec:governing}

In this section, we briefly describe the dimensional governing equations for flow of the foulant-laden feed solution and transport and deposition of foulant in a single pore, then extend the description to a network of such pores using conservation laws at pore junctions. For more details of the derivation, we refer the reader to Gu {\it et al.}~\cite{gu_network_2021}. 

\subsubsection{Fluid Flow}
The feed is assumed to be a Newtonian fluid with viscosity $\mu$, driven through a cylindrical pore with small aspect ratio via a pressure difference $P_0$. The flux $Q_{ij}$ between junctions $i$ and $j$ is characterized by the Hagen-Poiseuille equation~\cite{probstein},
\begin{equation}
    Q_{ij}=\mathbb{K}_{ij}\left(P_i - P_j\right), \quad e_{ij}\in \mathcal{E},
    \label{eq:HP}
\end{equation}
where $\mathbb{K}_{ij}$ is the conductance of the pore $e_{ij}$, given by 
\begin{equation}
\mathbb{K}_{ij}=\begin{cases}
\frac{\pi R_{ij}^{4}}{8\mu A_{ij}}, & e_{ij}\in \mathcal{E},\\
0, & \text{otherwise},
\end{cases}
\label{eq:grad_conductance}
\end{equation}
where $R_{ij}$ and $A_{ij}$ are the radius and length of $e_{ij}$. Enforcing conservation of flux at pore junctions leads to a system of equations for the pressure at each interior junction, subject to the boundary conditions
\begin{equation}
    P_{i}=
    \begin{cases}
    P_{0}, & i\in \mathcal{V}_{{\rm in}},\\
    0, & i\in \mathcal{V}_{{\rm out}},
    \end{cases}
\end{equation}
where $\mathcal{V}_{{\rm in}}$ and $\mathcal{V}_{{\rm out}}$ are the set of membrane pore inlets and outlets, respectively (see \cref{eq:grad_vtop_vbot} for their precise definitions). Once the pressures are known, we use \cref{eq:HP} to find the flux in each individual pore. The scales for these equations are defined in \cref{sec:grad_scales}.

\subsubsection{Foulant Transport}
In this work only adsorptive fouling is considered, in which foulant particles much smaller than the pore radius are transported by the flow and adhere to the pore wall due to a variety of chemical or physical effects that depend on the affinity between the particles and the membrane material. $C_{ij}$, the concentration of foulant particles in pore $e_{ij}$, satisfies the steady state advection equation,
\begin{equation}
    Q_{ij}\frac{\partial C_{ij}}{\partial Y}=-\Lambda R_{ij}C_{ij},
    \label{eq:advection}
\end{equation}
where $\Lambda$ is an affinity parameter that describes the interaction between foulant particles and the membrane material; and $Y$ is a local coordinate along the pore in the direction of flux $Q_{ij}$. \cref{eq:advection} is paired with a boundary condition
\begin{equation}
    C_{ij}\left(0,T\right)=\begin{cases}
C_{0}, & i\in \mathcal{V}_{{\rm in}},\\
C_{i}\left(T\right), & \text{otherwise},
\end{cases}
\end{equation}
where $C_{i}\left(T\right)$ is the (unknown) concentration at junction $i$, determined by enforcing conservation of foulant particle flux, $Q_{ij}C_{ij}$, at each junction $i$ (for all adjacent $j$). 

\subsubsection{Pore Radius Evolution}
Once we obtain the foulant concentration at each junction, we model pore-size evolution with a decay rate directly proportional to the concentration at the upstream inlet, {\it i.e.}, the pore junction with higher pressure (for a detailed justification of this model choice, which assumes pores remain cylindrical throughout their evolution, see \cite{gu_network_2021}). More precisely, $R_{ij}$, the radius of pore $e_{ij}$, satisfies
\begin{equation}
    \frac{dR_{ij}}{dT}=-\Lambda\alpha {C}_{i},\quad \alpha = \frac{V_p}{2\pi}, \quad e_{ij}\in \mathcal{E},
    \label{eq:adsorption}
\end{equation}
where $V_p$ is the effective volume of each foulant particle. This ODE is subject to the initial condition 
\begin{equation}
    \quad R_{ij}\left(0\right) = R_{k}, \quad e_{ij}\in \mathcal{E}_k, \quad k = 1,\ldots,m,
    \label{eq:radius_initial_condition}
\end{equation}
that is, we assign an initial pore radius $R_k$ to each pore that lies in the $k^{\rm th}$ band.

\section{Performance Metrics}\label{sec:perf_metrics}
We now define the performance metrics used in this work.

\begin{enumerate}
    \item Membrane Lifetime, $T_{\rm final}$. Two possible criteria defining membrane filter lifetime are considered:
        \begin{enumerate}
            \item (Flux extinction) $T_{\rm final}$ is the earliest time at which there exists no path from the top surface to the bottom. This corresponds to the time when flux is zero.
            \item (Flux threshold) $T_{\rm final}$ is the time at which flux level reaches a prescribed lower threshold. 
        \end{enumerate}
    \item Filtrate Throughput $H\left(T\right)$,
    \begin{align}
        H\left(T\right)=\int_{0}^{T}Q_{\rm out}\left(T^{\prime}\right)dT^{\prime}, \\
        Q_{{\rm out}}\left(T\right)=\sum_{v_{j}\in \mathcal{V}_{\rm out}}\sum_{i:e_{ij}\in \mathcal{E}}{Q}_{ij}\left(T\right), \label{eq:q_out}
    \end{align}
    where $Q_{{\rm out}}\left(T\right)$ is the total flux exiting the filter and $\mathcal{V}_{\rm out}$ is the set of pore outlets at the membrane bottom surface (see its detailed definition in \cref{eq:grad_vtop_vbot}). In particular, we are interested in $H_{\rm final}:=H\left(T_{\rm final}\right)$, the total volume of filtrate processed by the filter over its lifetime. 
    \item Initial Flux $Q_{\rm out}\left(T=0\right)$.
    \item Accumulated Concentration of Foulant at Membrane Outlet $C_{{\rm aco}}\left(T\right)$,
    \begin{equation}
       C_{{\rm aco}}\left(T\right)=\frac{\int_{0}^{T}C_{{\rm out}}\left(T^{\prime}\right)Q_{{\rm out}}\left(T^{\prime}\right)dT^{\prime}}{\int_{0}^{T}Q_{{\rm out}}\left(T^{\prime}\right)dT^{\prime}} = \frac{\int_{0}^{T}C_{{\rm out}}\left(T^{\prime}\right)Q_{{\rm out}}\left(T^{\prime}\right)dT^{\prime}}{H\left(T\right)}, 
       \label{eq:aco}
    \end{equation}
    where 
    \begin{equation}
    C_{{\rm out}}\left(T\right)=\frac{{\displaystyle \sum_{v_{j}\in \mathcal{V}_{\rm out}} \sum_{i:e_{ij}\in \mathcal{E}}}{C}_{j}\left(T\right){Q}_{ij}\left(T\right)}{Q_{{\rm out}}\left(T\right)} \label{eq:c_out}
    \end{equation}
    is the instantaneous foulant concentration at the membrane outlet. Of particular interest is $C_{\rm final}:=C_{{\rm aco}}\left(T_{\rm final}\right)$, which provides a measure of the purity of the total volume of filtrate collected over the filter lifetime (assuming a batch process). 
    
    \item Band porosity $\Phi_k$ and membrane porosity $\Phi$ as functions of time per~\cref{eq:band_porosity,eq:membrane_porosity} respectively, and their changes over the filter lifetime, referred to as {\bf band and membrane porosity usage} respectively,
    \begin{subequations}
        \begin{align}
            \Delta \Phi_k &= \Phi_k\left(0\right) - \Phi_k\left(T_{\rm final}\right), \quad k=1,\ldots,m, \label{eq:delta_band_porosity}\\
            \Delta \Phi &= \Phi\left(0\right) - \Phi\left(T_{\rm final}\right). \label{eq:delta_membrane_porosity}
        \end{align}
    \end{subequations}
    Per~\cref{eq:constraint_2a}, $\Phi_k\left(0\right) \approx \Phi\left(0\right) \approx 0.6$ where $\approx$ is up to some tolerance $\epsilon$.
    
\end{enumerate}

\section{Algorithm}\label{sec:algorithm}
We summarize the procedures described in \cref{sec:const_radius} in \cref{al:algorithm1}. 
\begin{algorithm}[!ht]
\caption{\bf Filtration on Pore-Size-Graded Networks}
\label{al:algorithm1}
\begin{enumerate}
    \item Choose maximum and minimum pore lengths $A_{\rm max}$, $A_{\rm min}$, average pore radius $R_0$ and porosity $\Phi$. \label{al:const_step1}
    \item Initialise radius gradient $s \geq 0$. Find $R_k$ as constrained by $R_0$ via \cref{eq:constraint_1a}. \label{al:const_step2}
    \item Generate $N_G$ banded networks parametrized by $A_{\rm max}$ and $R_k$:
        \begin{enumerate}
            \item Guess $N_m$ such that $\Phi_m \approx \Phi$.
            \item Determine $N_k$ for $k=1,\ldots,m-1$ via \cref{eq:num_k_relationship_2}.
            \item Generate networks using $N_k$ random junctions in the $k^{\rm th}$ band, connected according to the metric defined in \cref{eq:connection_metric}.
            \item Correct membrane porosity by adding or deleting junctions until it is within relative tolerance $\epsilon = 0.005$ of $\Phi$. More precisely, we perform this procedure until $ \left|1-\frac{\Phi_{\rm corr}}{\Phi}\right|<\epsilon$ where $\Phi_{\rm corr}$ is the porosity during this iterative procedure. \label{al:correction}
        \end{enumerate}    
    \item Compute the performance metrics (defined in \cref{sec:perf_metrics}).
    \item Go back to \cref{al:const_step2} by varying $s$.
    \item Go back to \cref{al:const_step1} by varying $A_{\rm max}$.
\end{enumerate}
\end{algorithm}

\section{Scales}\label{sec:grad_scales}

We nondimensionalize the model and key quantities introduced in \cref{sec:modelling}, and the performance metrics defined in \cref{sec:perf_metrics}, with the following scales,
\begin{equation}\label{eq:grad_scaling}
    \begin{gathered}
    P=P_{0}p,\qquad \left(A_{ij},L_{k,ij}\right)=W\left(a_{ij},l_{k,ij}\right),\\
    \left(A_{\rm min},A_{{\rm max}}\right)=W\left(a_{\rm min},a_{{\rm max}}\right),\qquad \left(R_{ij},R_0\right)=W\left(r_{ij},r_0\right),\\
    Q_{ij}=\frac{\pi W^{3}P_{0}}{8\mu}q_{ij},\qquad \mathbb{K}_{ij}=\frac{\pi W^{3}}{8\mu}\mathbb{k}_{ij},\qquad\mathbb{k}_{ij}=\frac{r_{ij}^{4}}{a_{ij}},\\
    \left(C_{ij},C_{\rm aco}\right) = C_0 \left(c_{ij},c_{\rm aco}\right),\qquad Z=Wz,\qquad\Lambda=\frac{\pi WP_{0}}{8\mu}\lambda,\qquad T=\frac{W}{\Lambda\alpha C_{0}}t.
    \end{gathered}
\end{equation}
\cref{eq:band_radius} for the initial pore radius in layer $k$, in dimensionless form now reads
\begin{equation}
    r_k = r_m + \left(m-k\right)s.
    \label{eq:dimless_band_radius}
\end{equation}
Using the scales for pore radius $R_{ij}$ and length $L_{k,ij}$ in the $k^{\rm th}$ band, band porosity becomes 
\begin{equation}
\Phi_k\left(t\right) := \frac{\frac{\pi}{2} \sum_{e_{ij}\in \mathcal{E}} r_{ij}^2\left(t\right) l_{k,ij}}{1/m}.
\label{eq:band_porosity_const_nondim}
\end{equation}
The dimensionless performance metrics become (with upper-case dimensional quantities replaced by their lower-case equivalents)
\begin{subequations}
    \begin{align}
    h\left(t\right)=\frac{1}{\lambda}\int_{0}^{t}q_{{\rm out}}\left(t^{\prime}\right)dt^{\prime},\quad & q_{{\rm out}}\left(t\right)=\sum_{v_{j}\in \mathcal{V}_{{\rm out}}}\sum_{v_{i}:\left(v_{i},v_{j}\right)\in \mathcal{E}}q_{ij}\left(t\right),\label{eq:dimless_tt}\\
    c_{{\rm aco}}\left(t\right)=\frac{\int_{0}^{t}c_{{\rm out}}\left(t^{\prime}\right)q_{{\rm out}}\left(t^{\prime}\right)dt^{\prime}}{\int_{0}^{t}q_{{\rm out}}\left(t^{\prime}\right)dt^{\prime}},\quad & c_{{\rm out}}\left(t\right)=\frac{{\displaystyle \sum_{v_{j}\in \mathcal{V}_{{\rm out}}}\sum_{v_{i}:\left(v_{i},v_{j}\right)\in \mathcal{E}}}c_{j}\left(t\right)q_{ij}\left(t\right)}{q_{{\rm out}}\left(t\right)},\label{eq:dimless_cacm}
    \end{align}
    \label{eq:dimless_metrics}
\end{subequations}

\cref{table:parameter_values,table:quantity} summarize the key nondimensional parameters and quantities, their symbols/definitions and range of values.

\begin{table}
\caption{\label{table:parameter_values}Key nondimensional parameters.}
\begin{ruledtabular}
\begin{tabular}{lll}
Parameter & Symbol & Values/Range \\
\midrule
Maximum pore length & $a_{\rm max}$ & $0.1,0.15,0.2$  \\
Minimum pore length & $a_{\rm min}$ & $0.06$  \\
Radius in $k^{\rm th}$ band & $r_k$ & $\left[2.5\times 10^{-3},0.016\right]$  \\
Uniform Radius (also average initial radius across the bands) & $r_0$ & $0.01$  \\
Radius (pore size) gradient & $s$ & $\left[0,4\times 10^{-3}\right]$  \\
Total number of bands & $m$ & $4$  \\
Deposition coefficient & $\lambda$ & $5\times 10^{-7}$  \\
Initial membrane porosity & $\Phi\left(0\right)$ & $0.6$ \\
Initial band porosity & $\Phi_k\left(0\right)$ & $0.6$ \\
Relative error for porosity correction (see \cref{al:algorithm1}) & $\epsilon$ & $0.005$
\end{tabular}
\end{ruledtabular}
\end{table}

\begin{table}
\caption{\label{table:quantity}Key nondimensional quantities.}
\begin{ruledtabular}
\begin{tabular}{lll}
Quantity & Symbol & Formula/Range of Values\\
\midrule
Pore Length & $a_{ij}$ & $\left[a_{\rm min}, a_{\rm max}\right]$ \\
Pore Radius & $r_{ij}\left(t\right)$ & $\left[2.5\times 10^{-3}, 0.016\right]$  \\
Pore length in the $k^{\rm th}$ band & $l_{k,ij}$ & see \cref{app:vertex_edge_defn} \\
Porosity of the $k^{\rm th}$ band & $\Phi_k\left(t\right)$ & $\frac{m\pi}{2}\sum_{e_{ij}\in \mathcal{E}}  r^2_{ij}\left(t\right)l_{k,ij}$ \\
Membrane Porosity & $\Phi\left(t\right)$ & $\frac{1}{m}\sum_{k=1}^m \Phi_k\left(t\right)$  \\
Total Throughput & $h\left(t\right)$ & $\frac{1}{\lambda}\int_{0}^{t}q_{\rm out}\left(t^{\prime}\right)dt^{\prime}$  \\
ACO (see \cref{eq:aco} for nomenclature)& $c_{\rm aco}\left(t\right)$ & $\frac{\int_{0}^{t}c_{{\rm out}}\left(t^{\prime}\right)q_{{\rm out}}\left(t^{\prime}\right)dt^{\prime}}{\lambda h\left(t\right)}$ 
\end{tabular}
\end{ruledtabular}
\end{table}

\section{Results and Discussions}\label{sec:grad_results}
In this section, we present results on the performance metrics of banded networks as the pore-radius gradient $s$ and maximum pore length $a_{\rm max}$ are varied. For each chosen value of $s$, we generate $1000$ networks independently using the banded network generation protocol described in \cref{sec:graded} and collect the mean and standard deviation of each performance metric defined in \cref{sec:perf_metrics}. The system of ODEs \cref{eq:adsorption} is solved using a simple forward Euler method with a time step size of $2.5\times 10^{-4}$. We investigate the trend of the mean performance metrics against $s$ and $a_{\rm max}$.

\subsection{Filter Performance until Flux Extinction}

\begin{figure}[!ht]
    \centering
    \subfloat[Total throughput]{\label{fig:tt}\includegraphics[width=.45\textwidth,height=.4\textwidth]{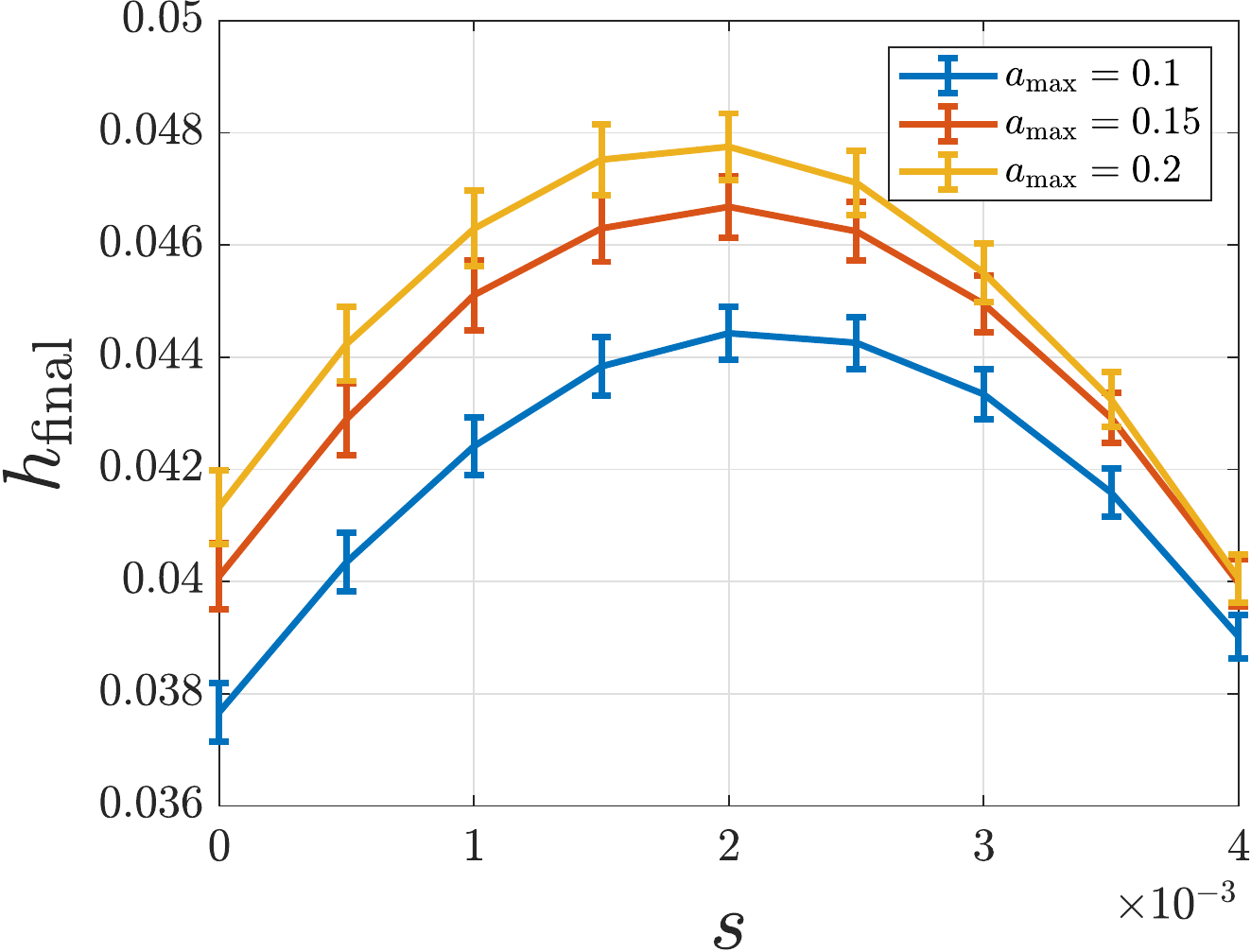}}\quad
    \subfloat[Initial flux]{\label{fig:q0}\includegraphics[width=.45\textwidth,height=.4\textwidth]{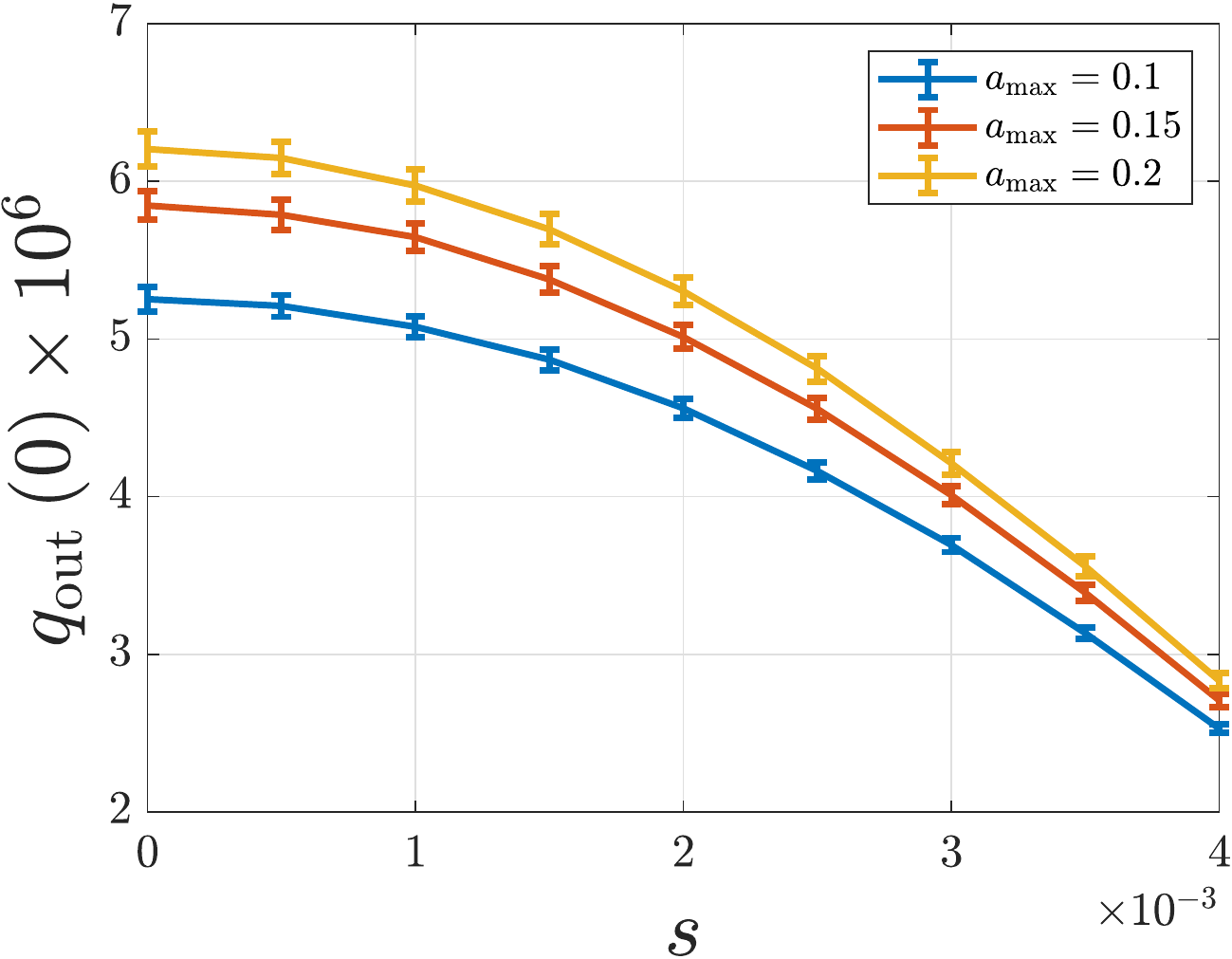}}\\
    \subfloat[Final accumulated foulant concentration]{\label{fig:c}\includegraphics[width=.45\textwidth,height=.4\textwidth]{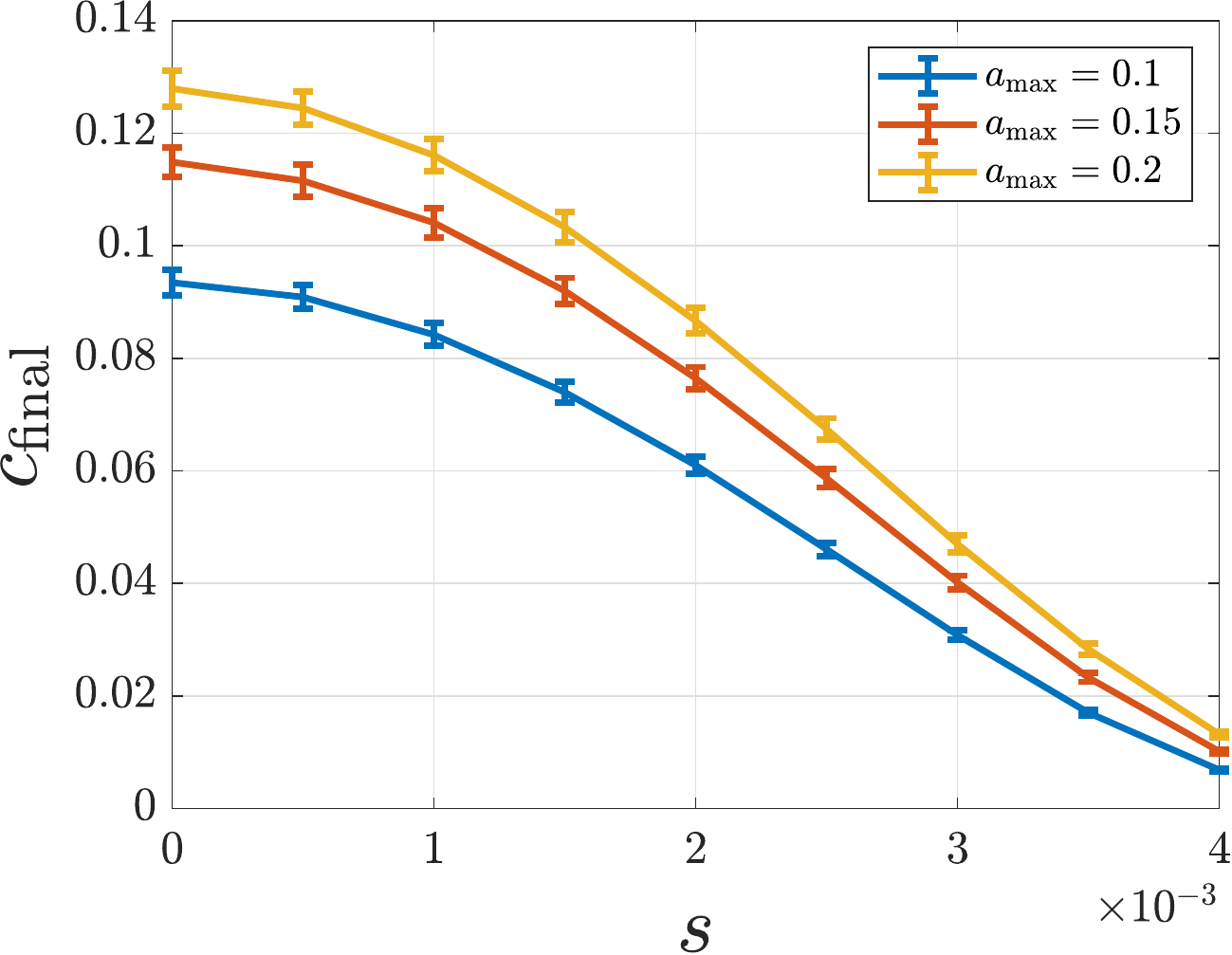}}\quad
    \caption{Performance metrics against radius gradient.}
    \label{fig:main}
\end{figure}

\cref{fig:main} presents how performance metrics vary with $s$ in banded networks with $m=4$ bands, within each of which the initial pore radius is constant (specified by \cref{eq:dimless_band_radius}). Results are shown for three different values of $a_{\rm max}$. \cref{fig:tt} shows total filtrate throughput against $s$. We observe that for each $a_{\rm max}$ value considered, we have a non-monotone trend with a clear maximum in total throughput at $s=2\times 10^{-3}$ (independent of $a_{\rm max}$, though the total filtrate throughput achieved in all cases is monotone increasing in $a_{\rm max}$). \cref{fig:q0} plots results for initial flux through the membrane, showing it to be a monotone decreasing function in $s$, and monotone increasing in  $a_{\rm max}$. \cref{fig:c} plots final accumulated foulant concentration at the membrane outlet against pore-size gradient $s$. The trend is monotone decreasing in $s$ for all $a_{\rm max}$ values considered (that is, more strongly graded networks provide better foulant control); and monotone increasing in $a_{\rm max}$.  

Before presenting further results, we first discuss the trends observed in \cref{fig:main}. We rationalize the existence of a throughput-maximizing value for $s$ by comparing networks with extreme values of $s$. Banded networks with large pore-size gradient tend to have low flux due to the high-resistance small outlets in the bottom band (evident in \cref{fig:q0} for all $a_{\rm max}$ values considered). At the same time, uniform networks ($s=0$) have the shortest lifetime due to their small pore inlets upstream (we hypothesize that pores always close first at the upstream side of the membrane, a claim supported by an analytical result in \cref{app:analytical}). Therefore, pore-size graded networks admit a trade-off between initial flux and filter lifetime as the gradient value varies, suggesting that an intermediate value should exist (and is found at $s=2\times 10^{-3}$ for the present choice of parameters) that maximizes total throughput.  Furthermore, the monotone trend of concentration against $s$ observed in \cref{fig:c} is not unexpected; the smaller pores in the bottom band of strongly graded networks are much more effective at removing foulant particles. 

We also briefly comment on the trend of the performance metrics as $a_{\rm max}$ varies in \cref{fig:main}. First, each metric (total throughput, initial flux and concentration) is monotonically increasing with $a_{\rm max}$. With the prescribed membrane porosity level $\Phi = 0.6$, these observed trends are consistent with the findings of Gu {\it et al.}~\cite{gu_network_2021}, in work that focused exclusively on networks with uniform pore radius (gradient $s=0$). Second, the reason for the higher foulant concentration from networks with larger $a_{\rm max}$ is that networks with longer pores tend to be less tortuous than those with shorter pores (smaller $a_{\rm max}$).\footnote[1]{The membrane {\bf tortuosity} is defined as the average (normalized) distance travelled by a fluid element from membrane top surface to bottom. See~\cite{gu_network_2021} for its detailed definition and discussion of the negative exponential relationship between concentration and tortuosity.}

\begin{figure}[!ht]
    \centering
    \subfloat[$a_{\rm max}=0.1$]{\label{fig:pore_evo_01}\includegraphics[width=.45\textwidth,height=.4\textwidth]{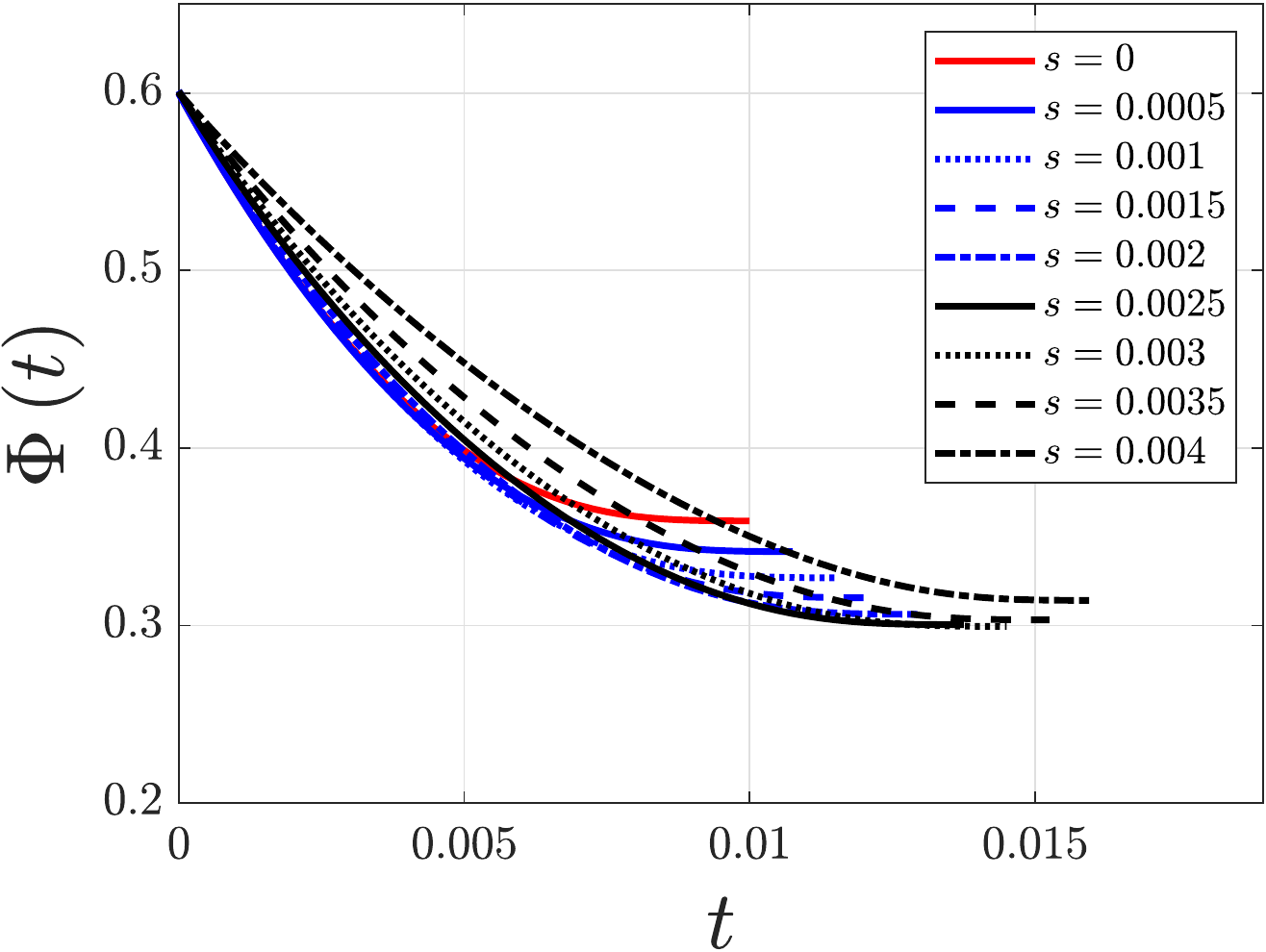}}\quad
    \subfloat[$a_{\rm max}=0.15$]{\label{fig:pore_evo_015}\includegraphics[width=.45\textwidth,height=.4\textwidth]{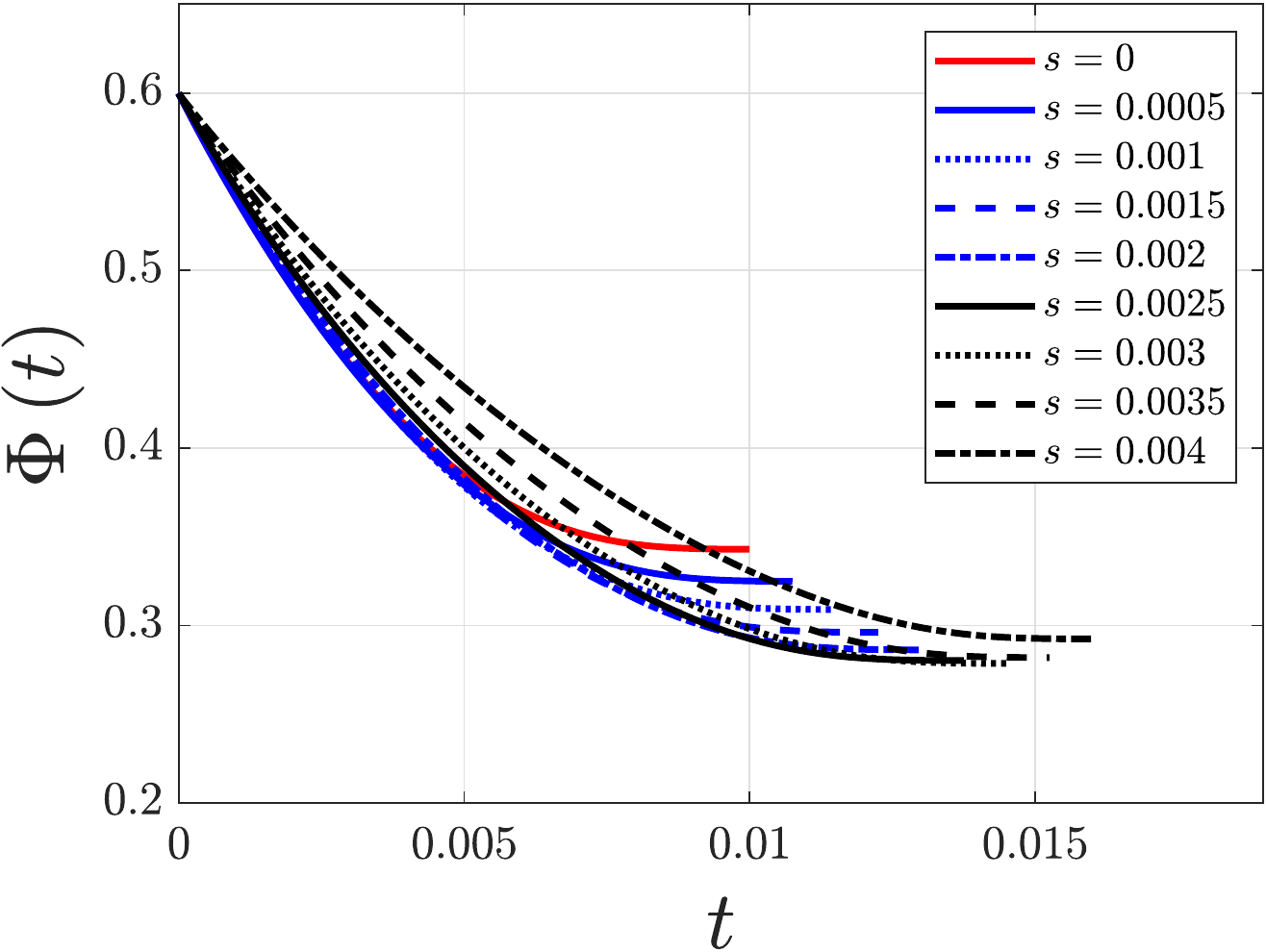}}\\
    \subfloat[$a_{\rm max}=0.2$]{\label{fig:pore_evo_02}\includegraphics[width=.45\textwidth,height=.4\textwidth]{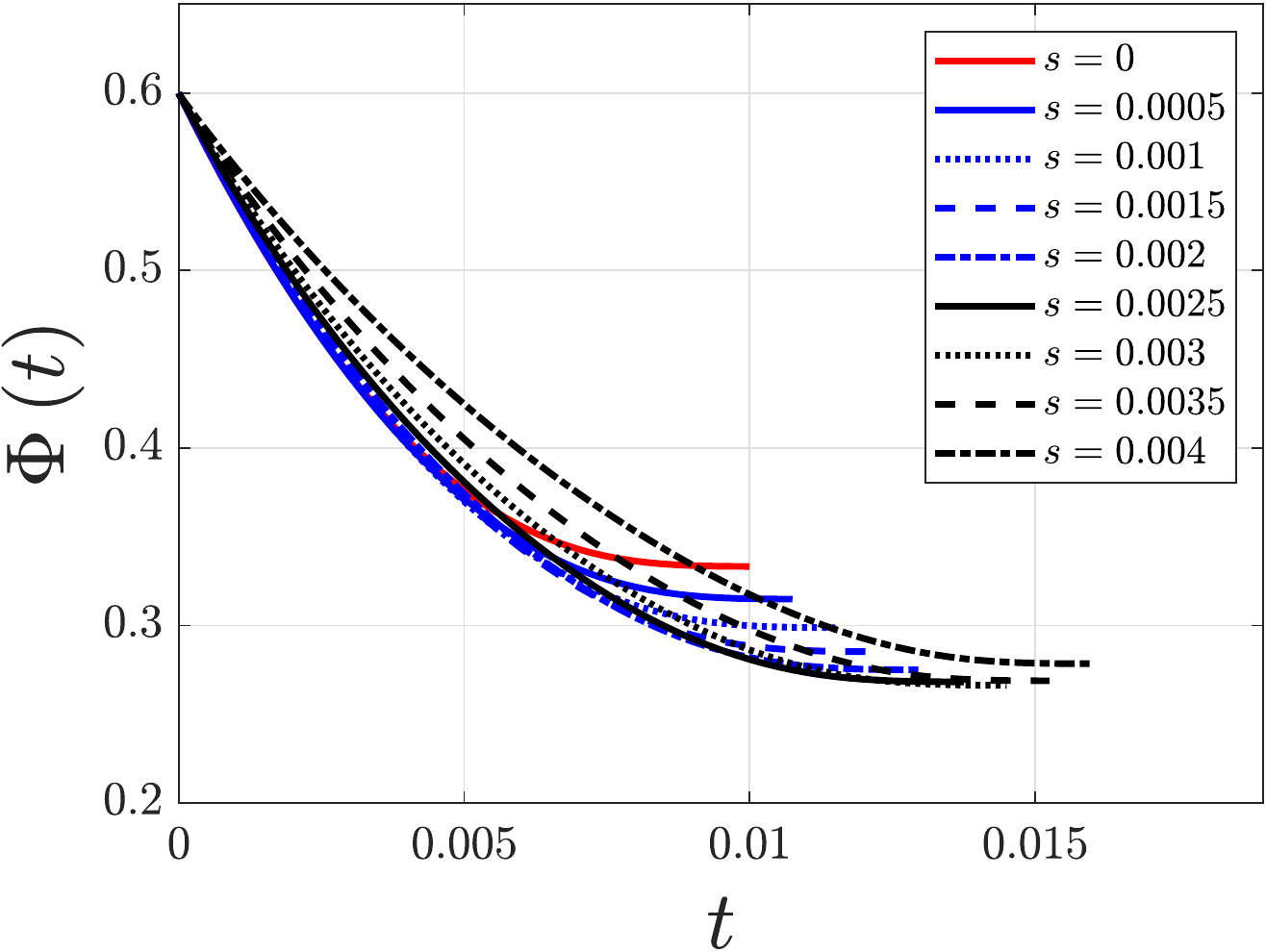}}\quad 
    \subfloat[Porosity usage]{\label{fig:pore_final}\includegraphics[width=.45\textwidth,height=.4\textwidth]{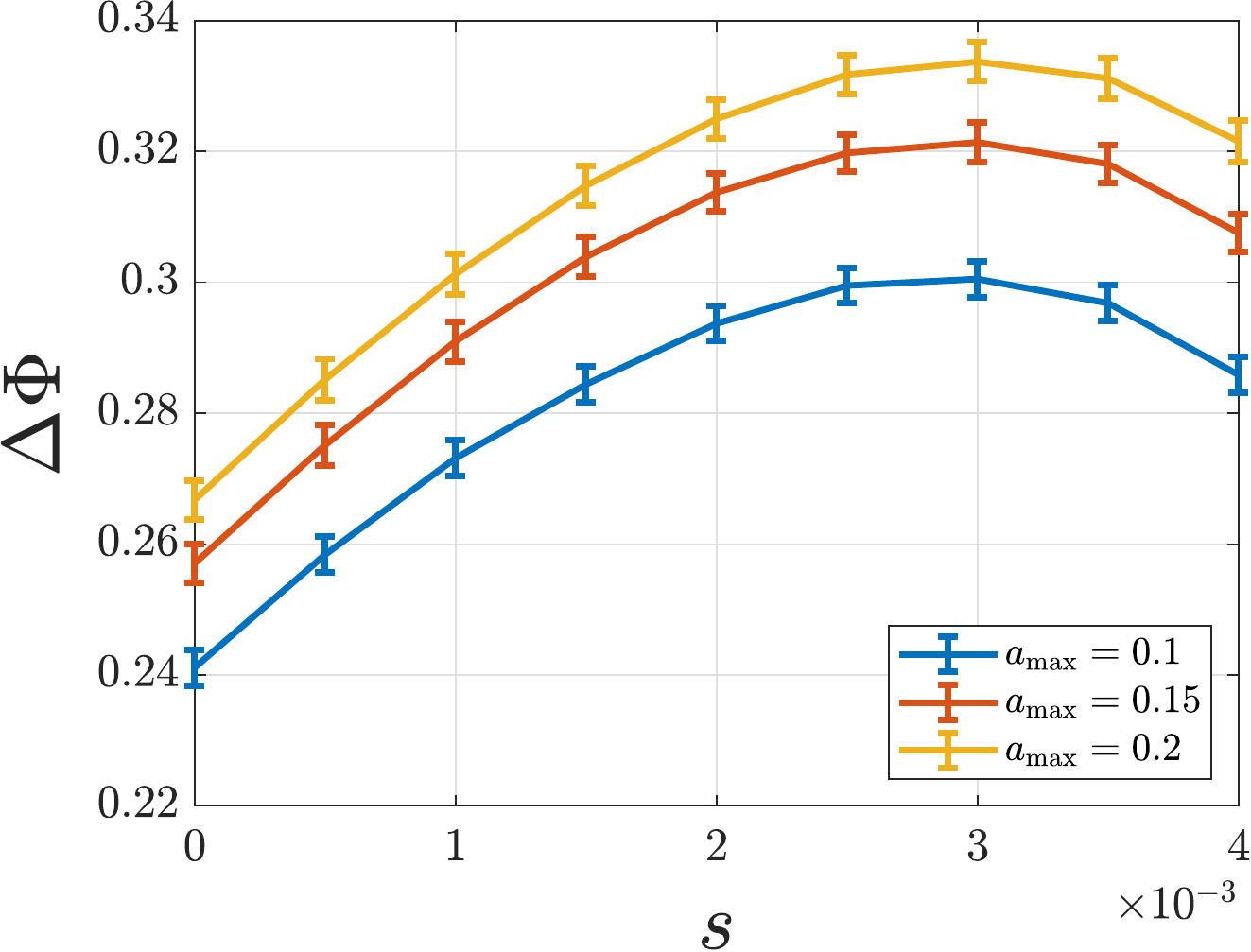}}
    \caption{(a-c) Total porosity evolution for different values of $a_{\rm max}$. In (d), membrane porosity usage, \cref{eq:delta_membrane_porosity}, is plotted against radius gradient $s$ for each $a_{\rm max}$. $\Phi\left(0\right) = 0.6$ is the initial porosity.}
    \label{fig:pore_evo}
\end{figure}

\subsubsection{Total Porosity Evolution}\label{sec:porosity_evo}

Porosity inevitably decreases over time as fouling occurs. In this subsection, we present how overall membrane porosity evolves in time as $a_{\rm max}$ and $s$ are varied. Our discussion focuses on the changes in membrane porosity and the final values achieved when filtration ceases, in particular the difference $\Delta\Phi$ between the initial $\Phi\left(0\right)=0.6$ and the final porosity, which we refer to as {\bf porosity usage} (see \cref{eq:delta_membrane_porosity}).

In each of \cref{fig:pore_evo_01,fig:pore_evo_015,fig:pore_evo_02} we show, for each $a_{\rm max}$, the evolution of membrane porosity in time, for all $s$-values considered. Filter lifetime may be inferred from the various curves by noting the time at which they stop (due to flux reaching zero). To showcase the porosity usages $\Delta\Phi$ of networks with different values of $s$, we condense them into \cref{fig:pore_final}, which clearly shows that $\Delta\Phi$ is a non-monotone function of $s$. In particular, for each $a_{\rm max}$, we find an optimal value of $s$ that incurs the largest porosity change. The figure also shows that networks with longer pores (the largest $a_{\rm max}$ value, \cref{fig:pore_evo_02}) incur the largest porosity changes over the filter lifetime. 

\subsubsection{Band Porosity Evolution}

While overall membrane porosity evolution shines light on the behavior of pore-radius-graded networks, individual band porosity evolution helps us identify the depth of foulant penetration in the membrane, for each pore-radius gradient value $s$ considered. In this subsection, we explore how band porosities change as $s$ varies and aim to draw further insight from this evolution into indicators of good pore-size-graded filters. The following discussion again focuses on the quantitative changes of band porosities, and the final porosity values when filtration stops.

\begin{figure}[!ht]
    \centering
    \subfloat[Band 1 porosity]{\label{fig:band1_evo_01}\includegraphics[width=.45\textwidth,height=.4\textwidth]{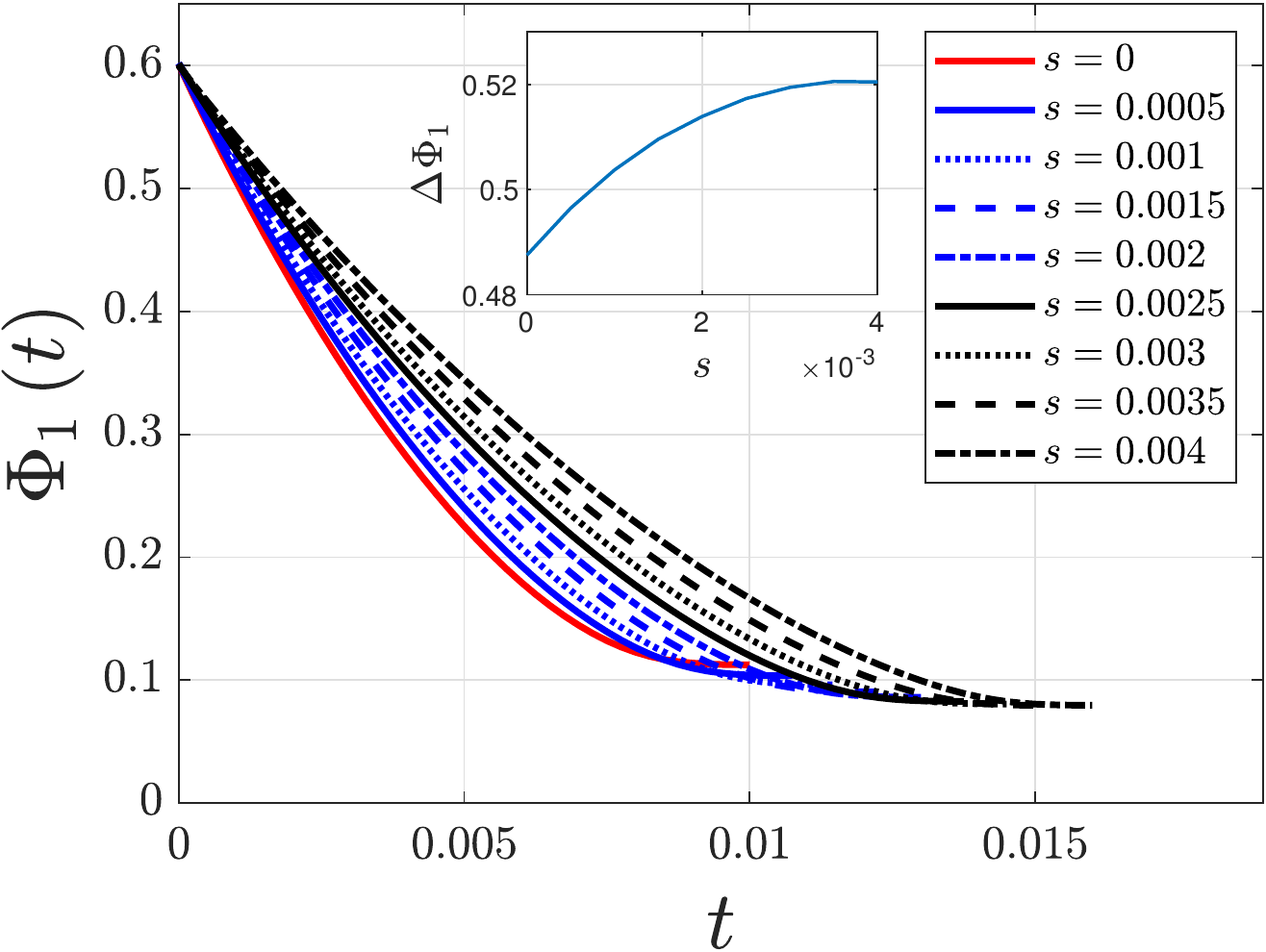}}\quad
    \subfloat[Band 2 porosity]{\label{fig:band2_evo_01}\includegraphics[width=.45\textwidth,height=.4\textwidth]{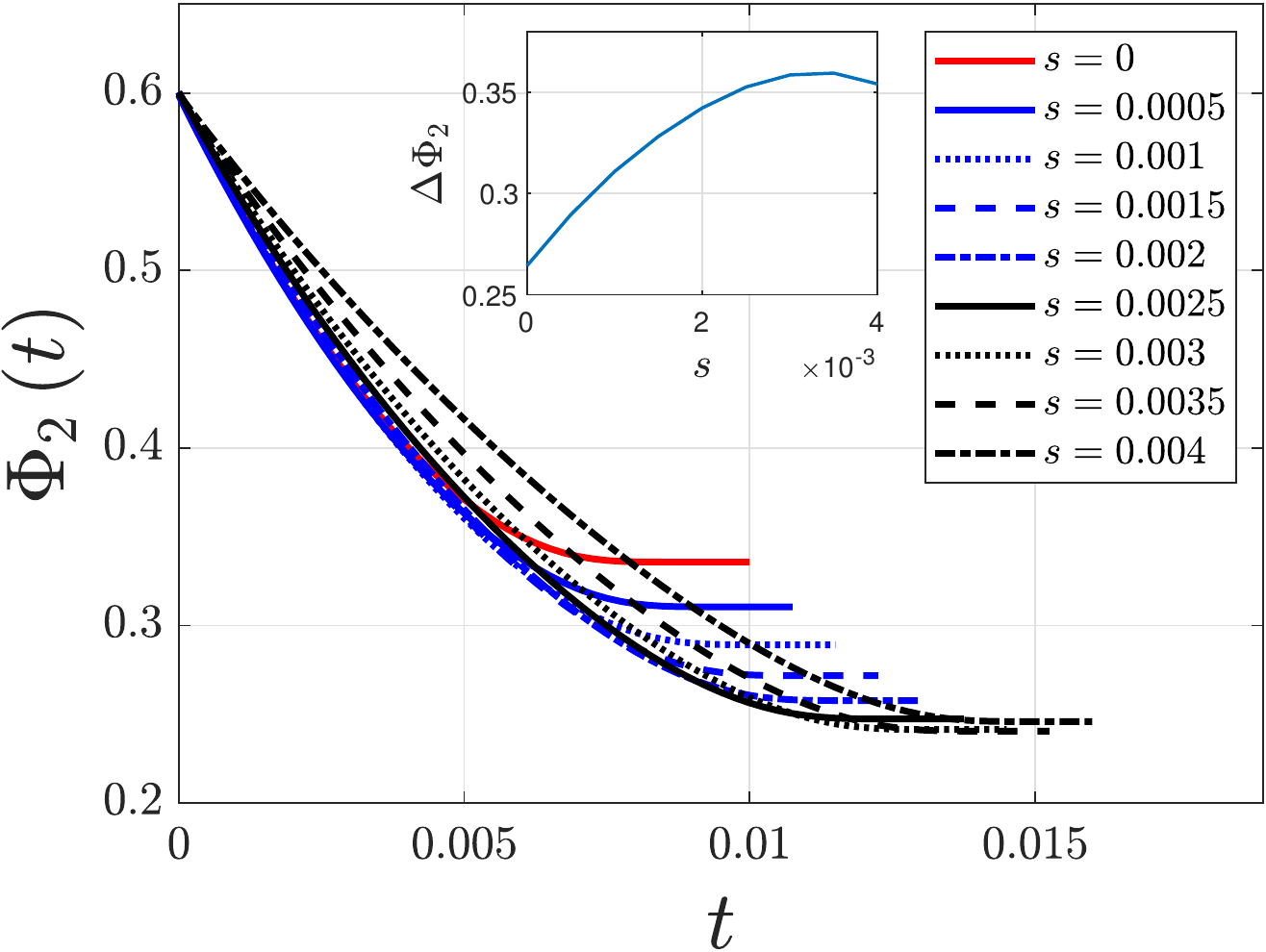}}\\
    \subfloat[Band 3 porosity]{\label{fig:band3_evo_01}\includegraphics[width=.45\textwidth,height=.4\textwidth]{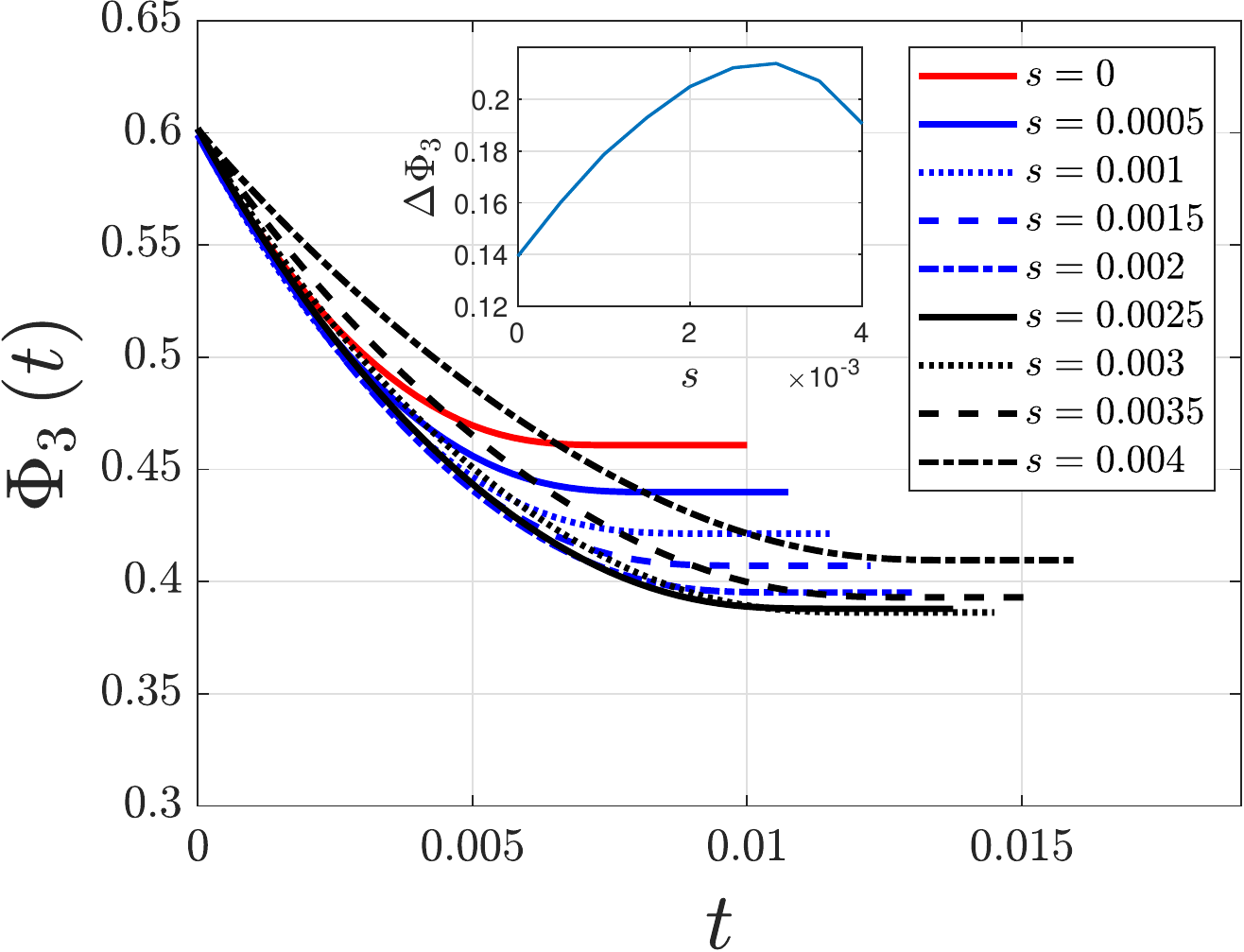}}\quad
    \subfloat[Band 4 porosity]{\label{fig:band4_evo_01}\includegraphics[width=.45\textwidth,height=.4\textwidth]{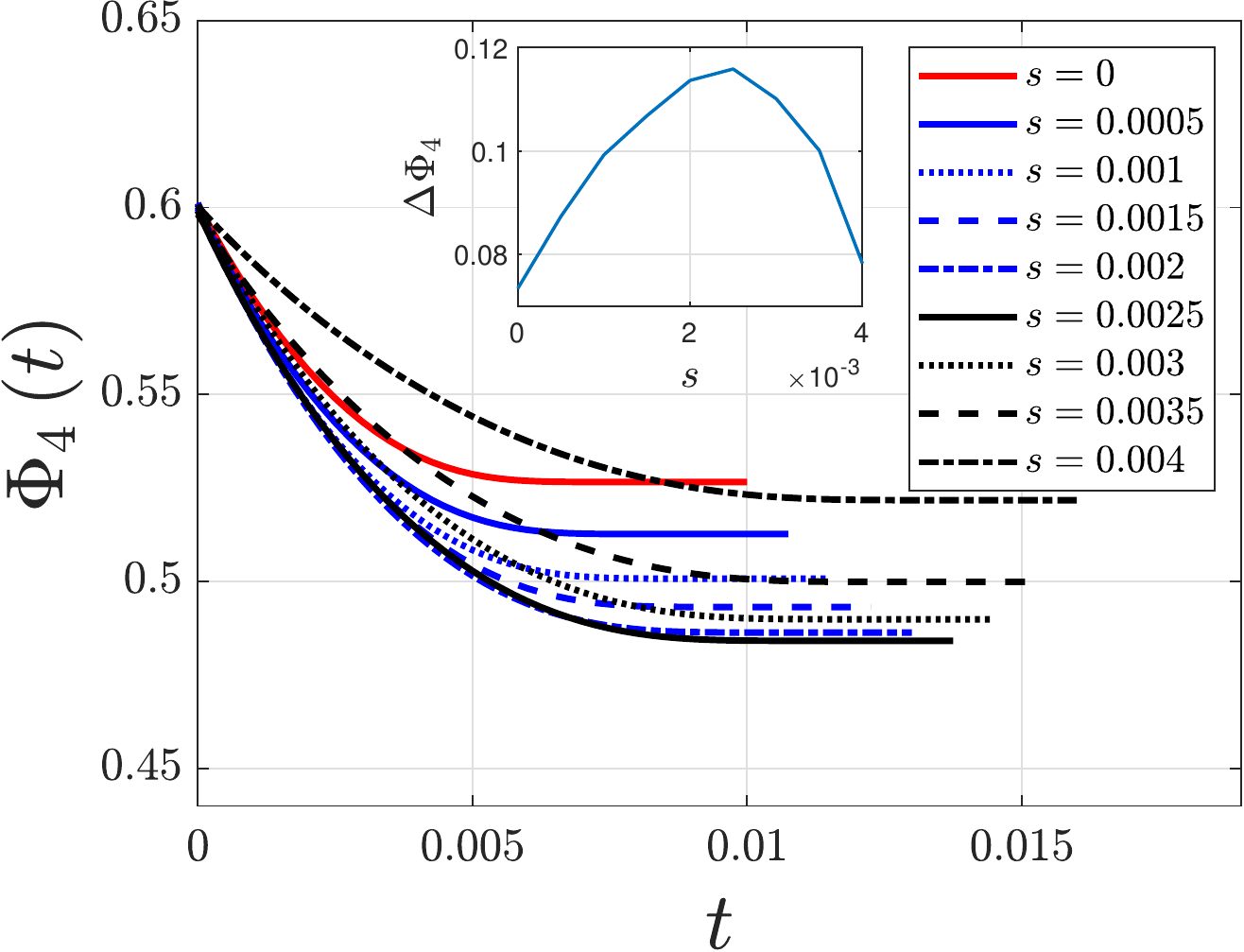}}
    \caption{Band porosity evolution with $a_{\rm max} = 0.1$. Inset subfigures plot band porosity usage (the total change in band porosity over the filter lifetime) as a function of radius gradient $s$. }
    \label{fig:band_evo_01}
\end{figure}

\begin{figure}[!ht]
    \centering
    \subfloat[Band 1 porosity]{\label{fig:band1_evo_015}\includegraphics[width=.45\textwidth,height=.4\textwidth]{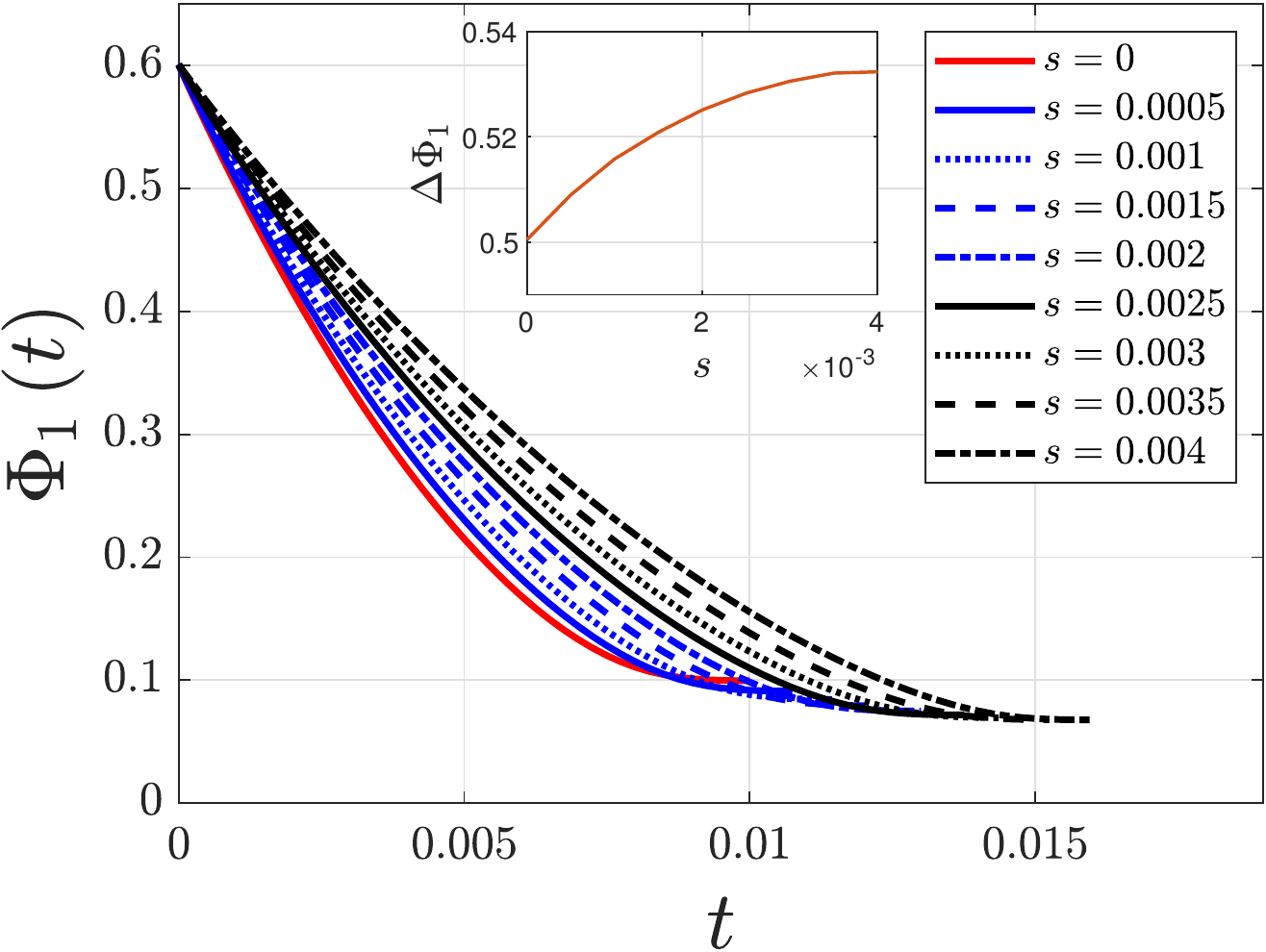}}\quad
    \subfloat[Band 2 porosity]{\label{fig:band2_evo_015}\includegraphics[width=.45\textwidth,height=.4\textwidth]{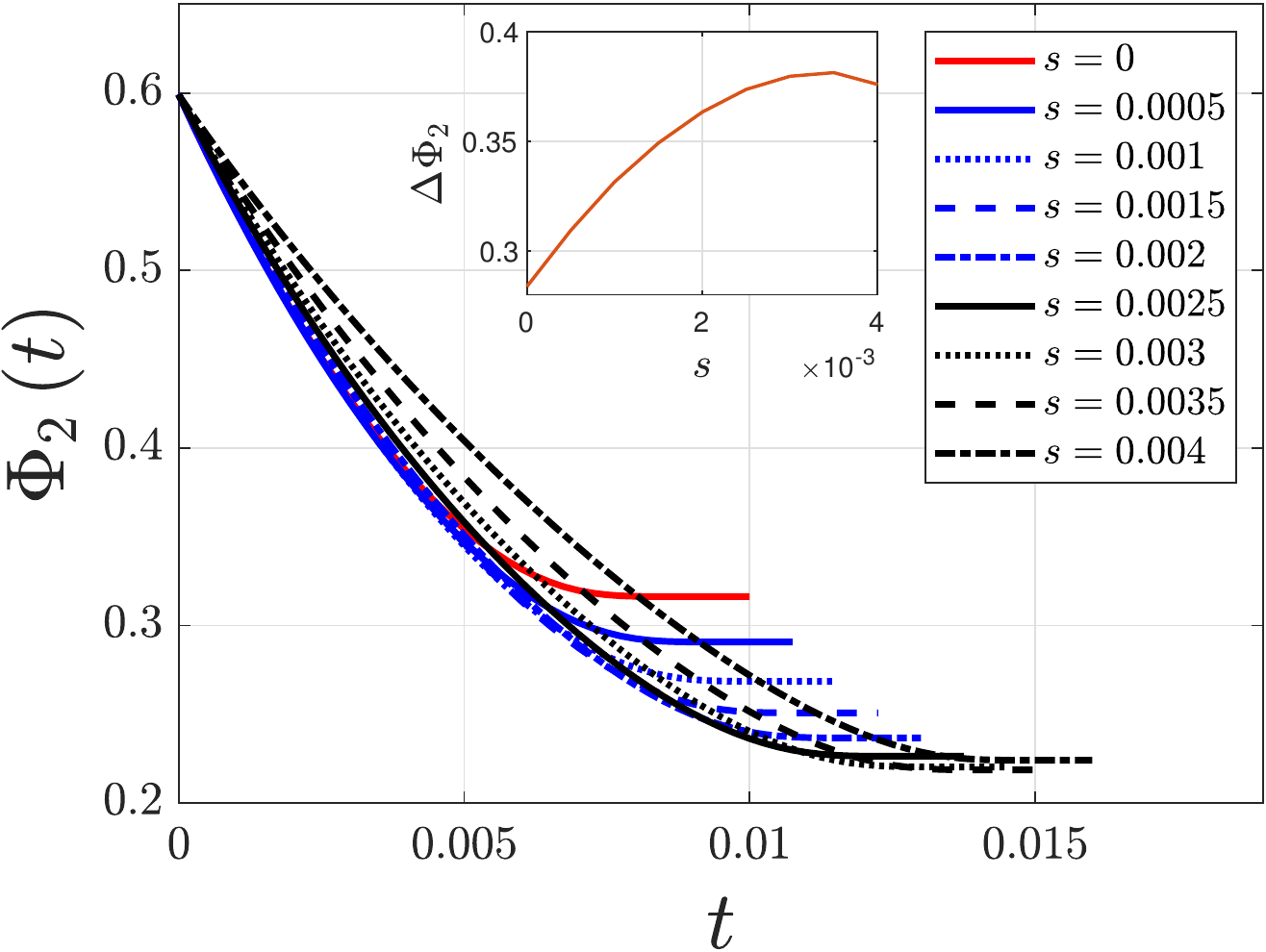}}\\
    \subfloat[Band 3 porosity]{\label{fig:band3_evo_015}\includegraphics[width=.45\textwidth,height=.4\textwidth]{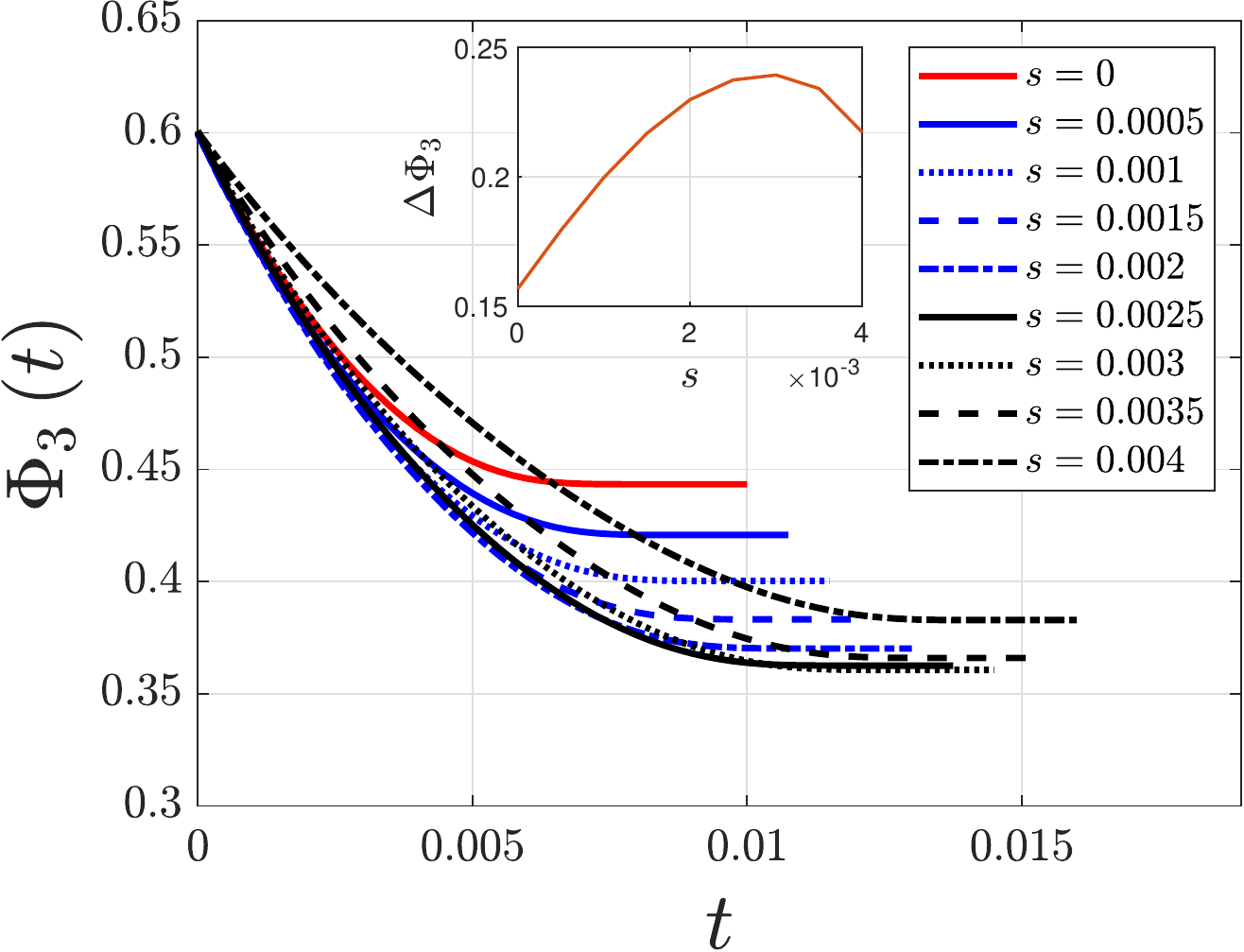}}\quad
    \subfloat[Band 4 porosity]{\label{fig:band4_evo_015}\includegraphics[width=.45\textwidth,height=.4\textwidth]{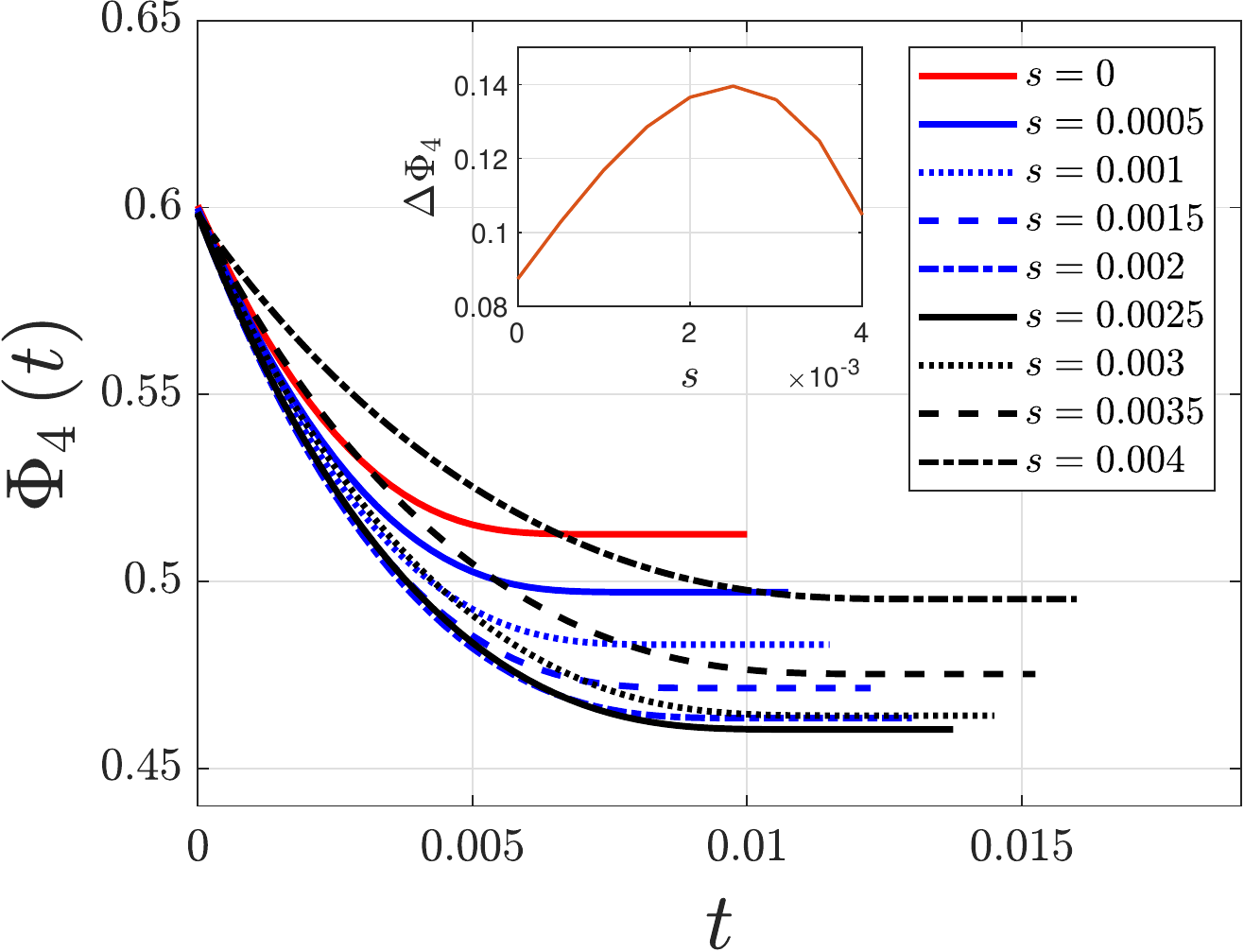}}
    \caption{Same setup as \cref{fig:band_evo_01}, with $a_{\rm max} = 0.15$.}
    \label{fig:band_evo_015}
\end{figure}

\begin{figure}[!ht]
    \centering
    \subfloat[Band 1 porosity]{\label{fig:band1_evo_02}\includegraphics[width=.45\textwidth,height=.4\textwidth]{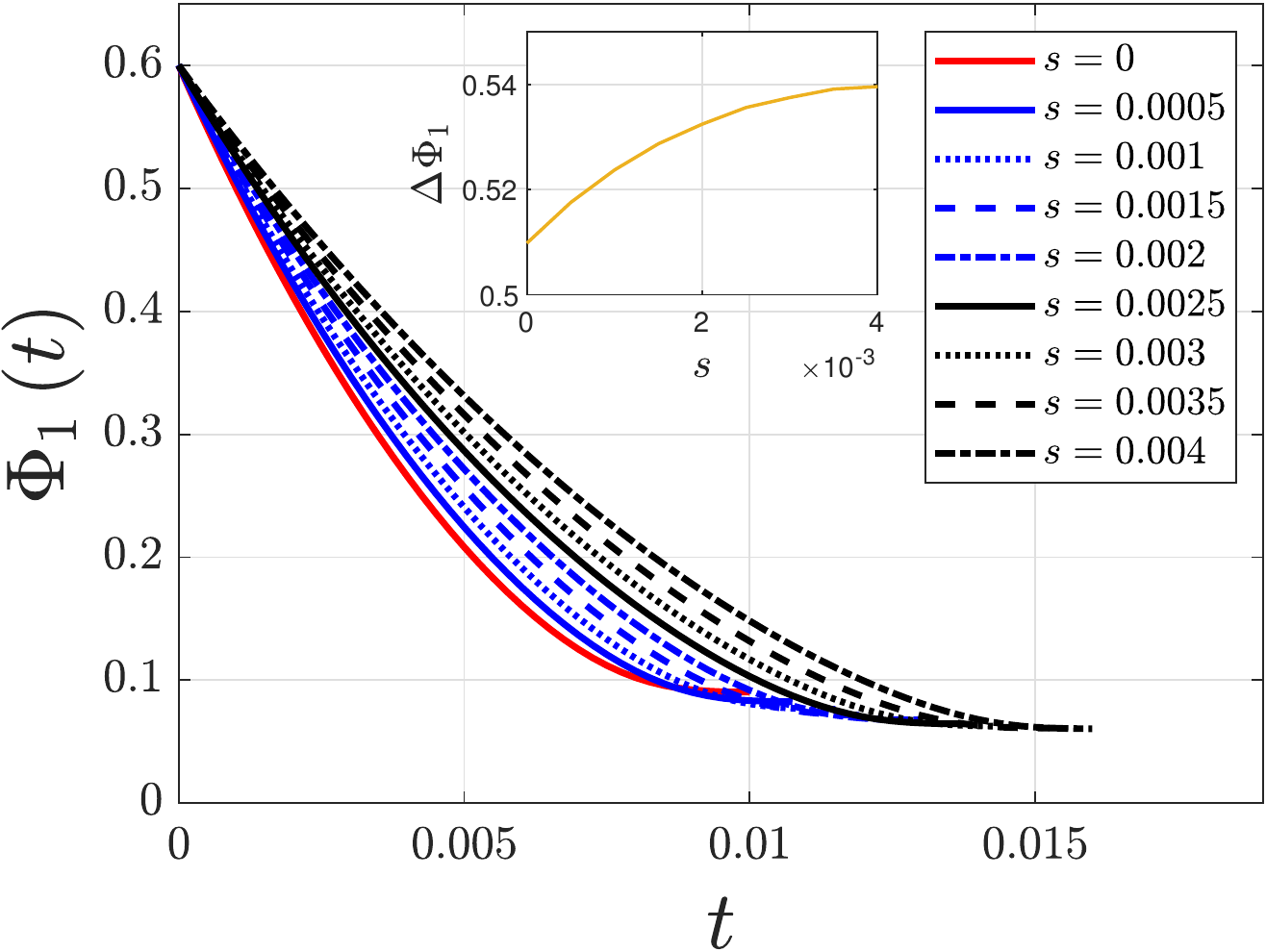}}\quad
    \subfloat[Band 2 porosity]{\label{fig:band2_evo_02}\includegraphics[width=.45\textwidth,height=.4\textwidth]{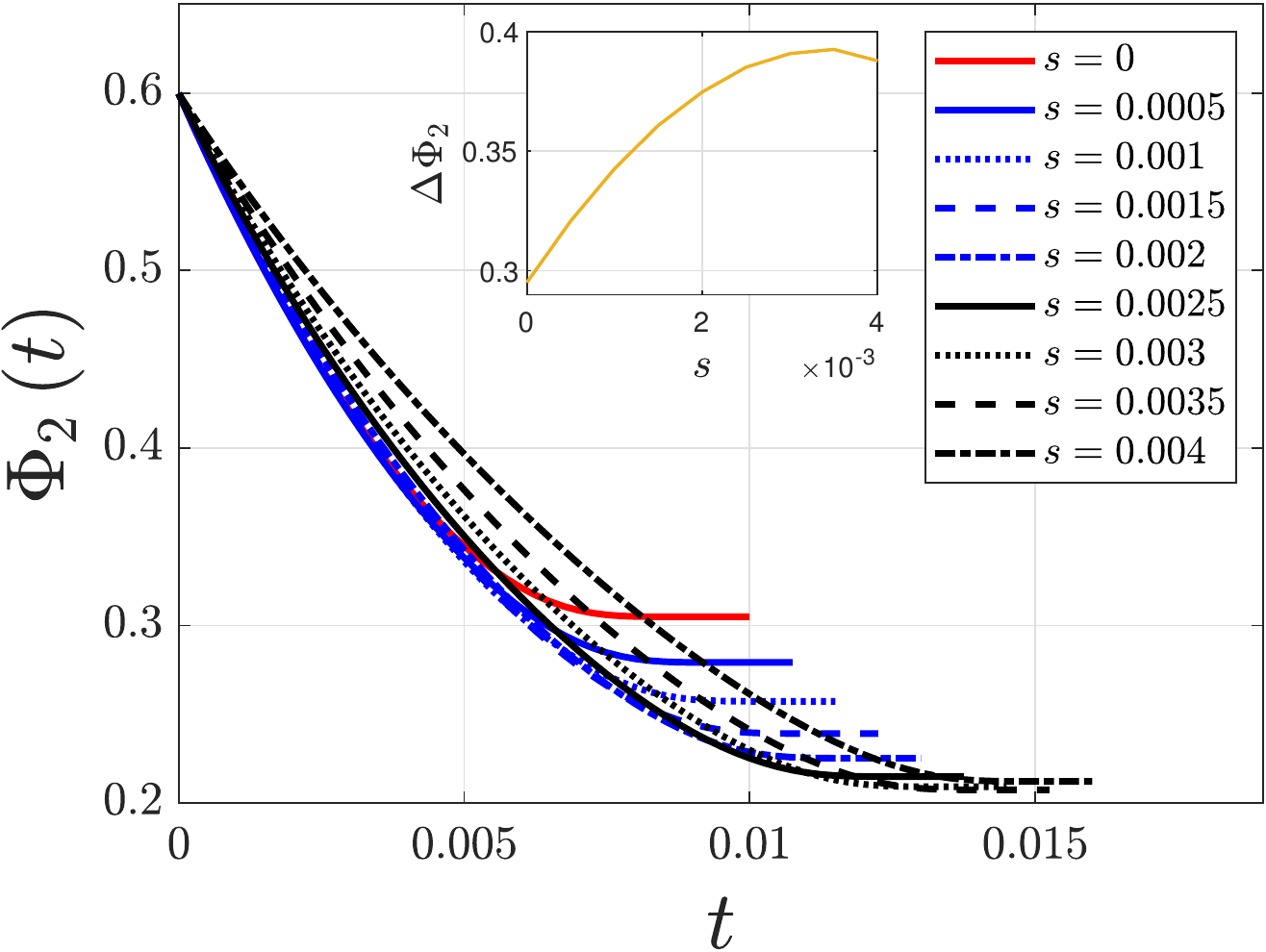}}\\
    \subfloat[Band 3 porosity]{\label{fig:band3_evo_02}\includegraphics[width=.45\textwidth,height=.4\textwidth]{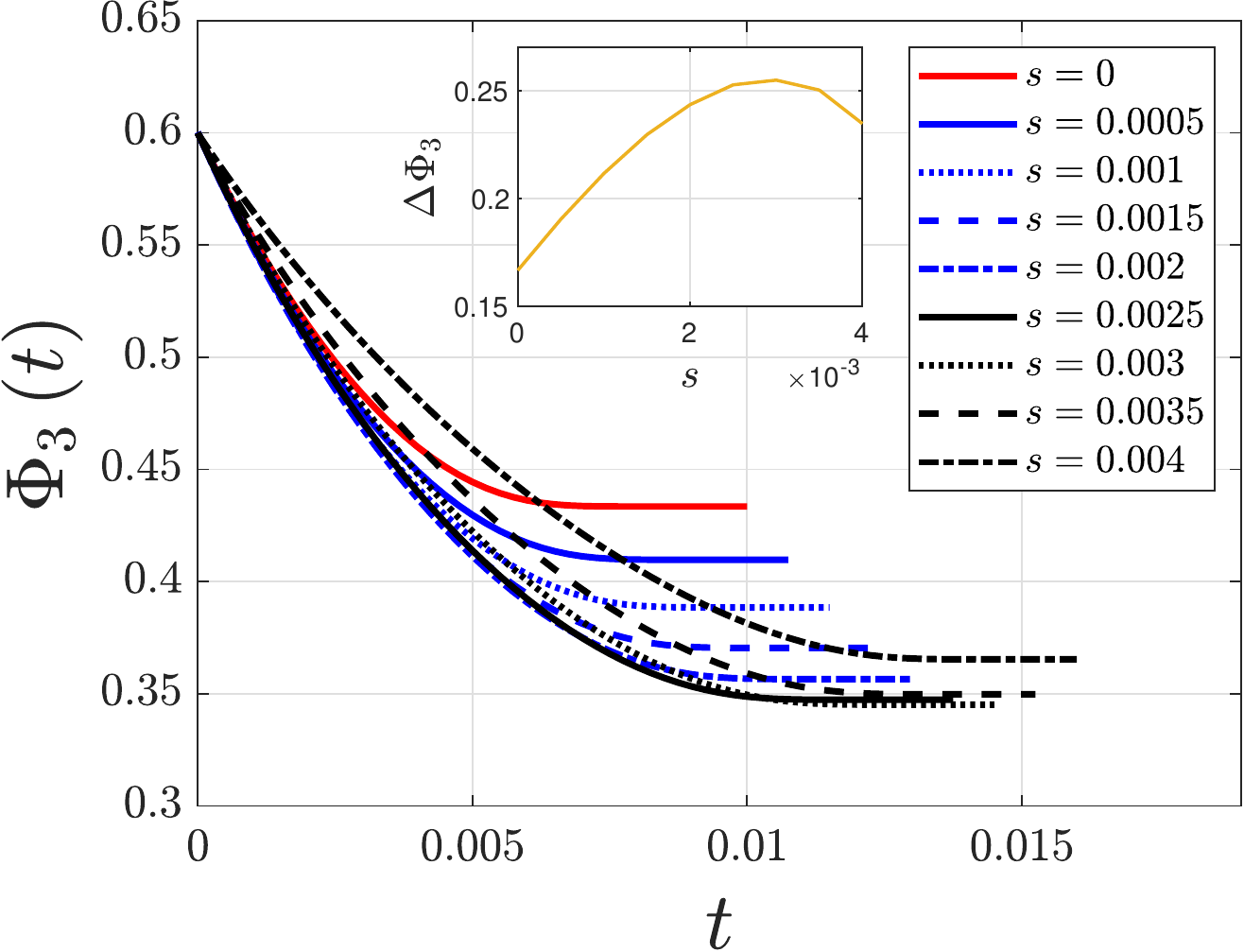}}\quad
    \subfloat[Band 4 porosity]{\label{fig:band4_evo_02}\includegraphics[width=.45\textwidth,height=.4\textwidth]{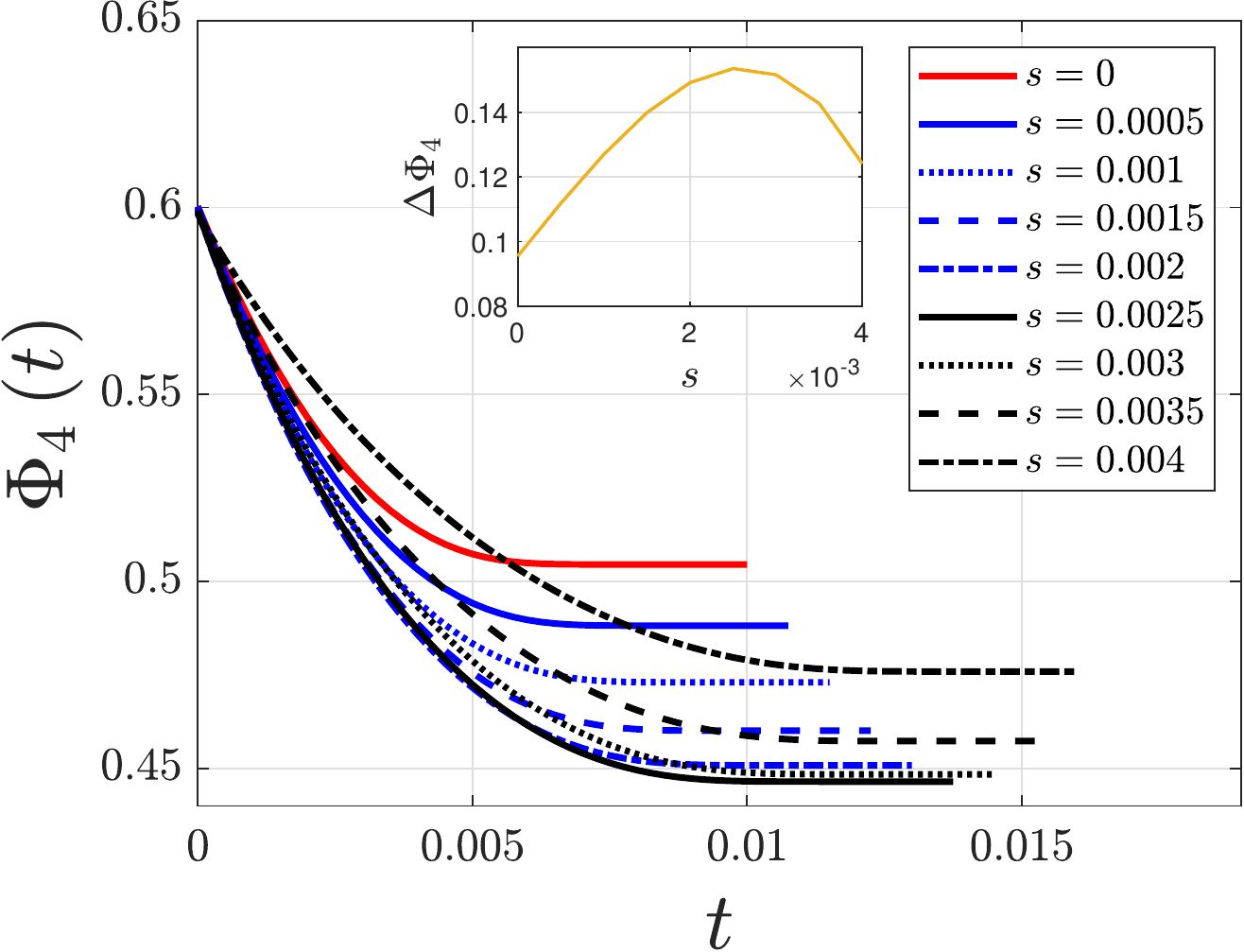}}
    \caption{Same setup as \cref{fig:band_evo_01}, with $a_{\rm max} = 0.2$.}
    \label{fig:band_evo_02}
\end{figure}

\cref{fig:band_evo_01} shows the evolution of band porosities $\Phi_k$ for the smallest $a_{\rm max}$-value considered (shortest pores, $a_{\rm max} = 0.1$), and for all radius gradient values $s$; the inset subfigure plots band porosity usage (the total change in band porosity, $\Delta\Phi_k$, see \cref{eq:delta_band_porosity}) as a function of $s$. In \cref{fig:band1_evo_01}, the evolution of the first band porosity (the upstream band) shows that the larger the value of $s$, the larger the porosity decline over the filter lifetime (see inset). This is because the largest $s$-value yields the largest first band pore radius and thereby the longest filter lifetime, allowing more particles to adsorb and thus using more empty space in the interior. We note in passing that, though the lifetimes of the networks (evidenced by the times at which each porosity curve stops) are quite different for different $s$-values, the final values of first band porosity are quite similar (exemplified by the small range of vertical axis in the inset), which implies that the first band processes foulants similarly regardless of the pore-size gradient $s$. \cref{fig:band2_evo_01} plots the evolution of the second band porosity against time for each $s$. Here, the final porosity values are clearly separated according to their $s$-values, in contrast to \cref{fig:band1_evo_01}. In particular, the uniform networks (in {\color{red}red}) clearly undergo the smallest band 2 porosity change over the filter lifetime, and thus retain the smallest mass of particles within this $2^{\rm nd}$ layer. The porosity usage of this band increases with $s$ until some value in the range $s\in [3\times 10^{-3}, 3.5\times 10^{-3}]$, for which the largest total change in band porosity is observed (see inset of \cref{fig:band2_evo_01}). In \cref{fig:band3_evo_01} where we show third band porosity evolution, the largest porosity change occurs in networks with $s=3\times 10^{-3}$, implying that foulant particles penetrate deeper into membrane pore networks with this gradient value. Lastly, we see in \cref{fig:band4_evo_01} that networks with radius gradient $s=2.5\times 10^{-3}$ (solid black) experience the largest change in fourth band porosity (corresponding to the maximum in the inset), indicating that such networks allow the deepest penetration of foulants, at least for the present choice of parameters. 

Meanwhile, membrane networks with $s=0$ (uniform pore size) and $s=4\times 10^{-3}$ (the steepest pore-size gradient) each perform relatively poorly in terms of porosity usage in the $4^{\rm th}$ band (their $4^{\rm th}$ band porosity does not change as appreciably as that of networks with other $s$-values). To explain this, in the case of uniform networks, their upstream pores close earlier than those of the graded counterparts (which have larger inlets due to the radius constraint via \cref{eq:band_radius,eq:constraint_1a}), thus prohibiting flow at an early stage by fouling upstream pores too quickly; in the case of $s=4\times 10^{-3}$ (the largest pore-size gradient used), the smaller downstream pores with their high resistance slow down the overall flow, causing the majority of fouling to take place upstream. 

In \cref{fig:band_evo_015,fig:band_evo_02}, following \cref{fig:band_evo_01}, we plot the band porosity evolution, for $a_{\rm max} = 0.15$ and $a_{\rm max} = 0.2$ respectively. We discover very similar correlations, between porosity usage and the fouling of downstream bands, to those just discussed for \cref{fig:band_evo_01} ($a_{\rm max} = 0.1$). The $1^{\rm st}$ band porosity change, $\Delta\Phi_1$, is always monotone increasing in pore-radius gradient $s$, and uniform networks incur the least porosity usage in their $4^{\rm th}$ band, indicating less deep penetration of foulants into the membrane and inefficient membrane usage (those with $s$-values that are too large also exhibit poor foulant penetration). Altogether, when taking total throughput, accumulated concentration of foulants and porosity usage all into account, we emphasize that an optimal value of radius gradient, largely independent of maximal pore length $a_{\rm max}$, exists. With the parameter values of Table \ref{table:parameter_values}, we find that membrane pore networks with $s=2.5\times 10^{-3}$ make the most favourable filters under the filtration strategy of flux exhaustion.

\subsection{Filter Performance with A Flux Threshold}

The results discussed so far are based on performance metrics evaluated at the end of the filter's lifetime, when there is no feasible flow path and flux falls to zero. In practice, when users observe a low flux level in the filtration process, they tend to discard the fouled filters and replace with fresh ones. In this section we mimic this procedure by imposing a minimal threshold for the flux level at which we halt the process and collect statistics of the performance metrics up to this critical time. The symbols for each performance metric $F$ evaluated with an imposed flux threshold are labelled with a subscript, $F_{\rm ths}$.

\begin{figure}[!ht]
    \centering
    \subfloat[Total throughput]{\label{fig:h_flux_threshold}\includegraphics[width=.45\textwidth,height=.4\textwidth]{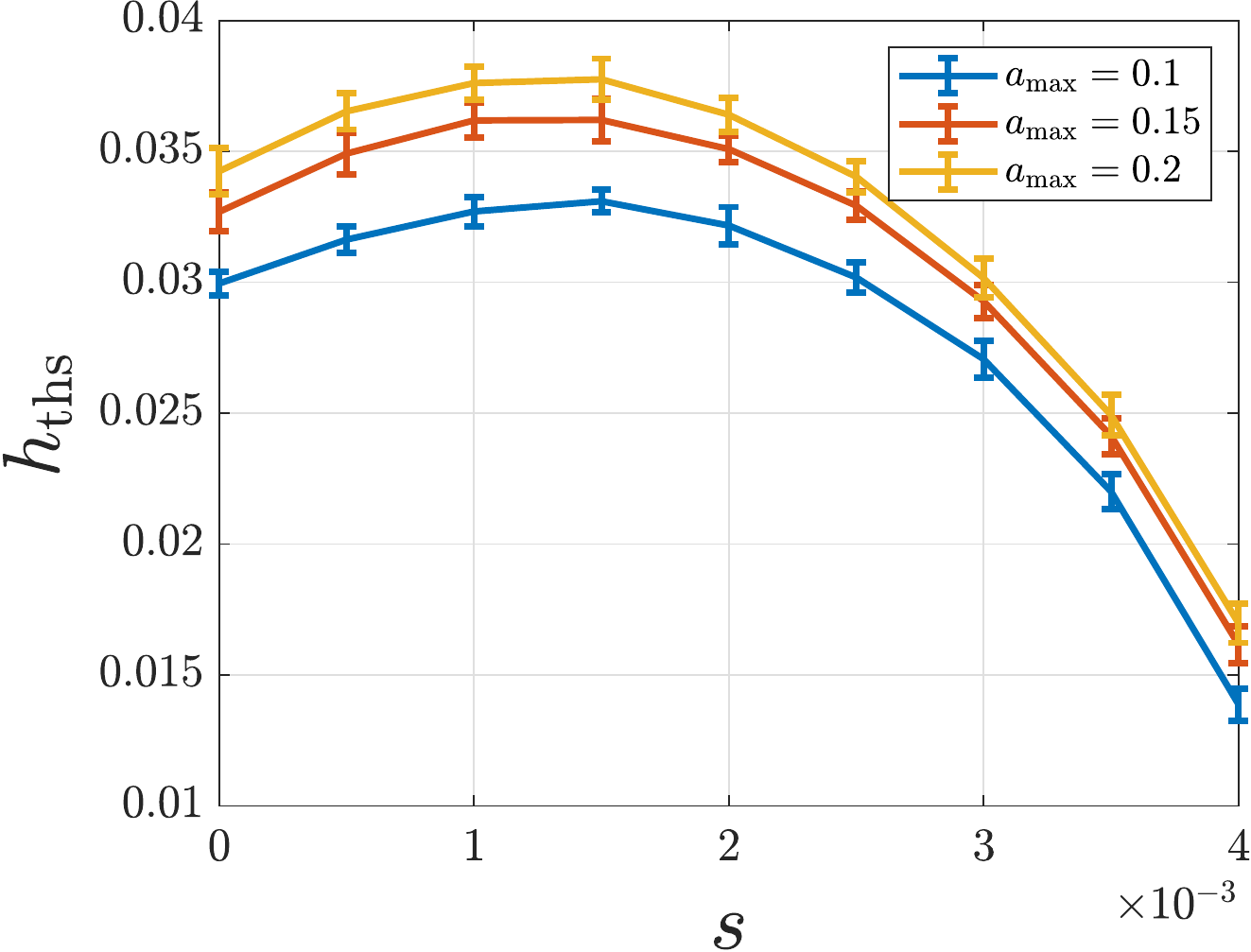}}\quad
    \subfloat[Final accumulated foulant concentration]{\label{fig:c_flux_threshold}\includegraphics[width=.45\textwidth,height=.4\textwidth]{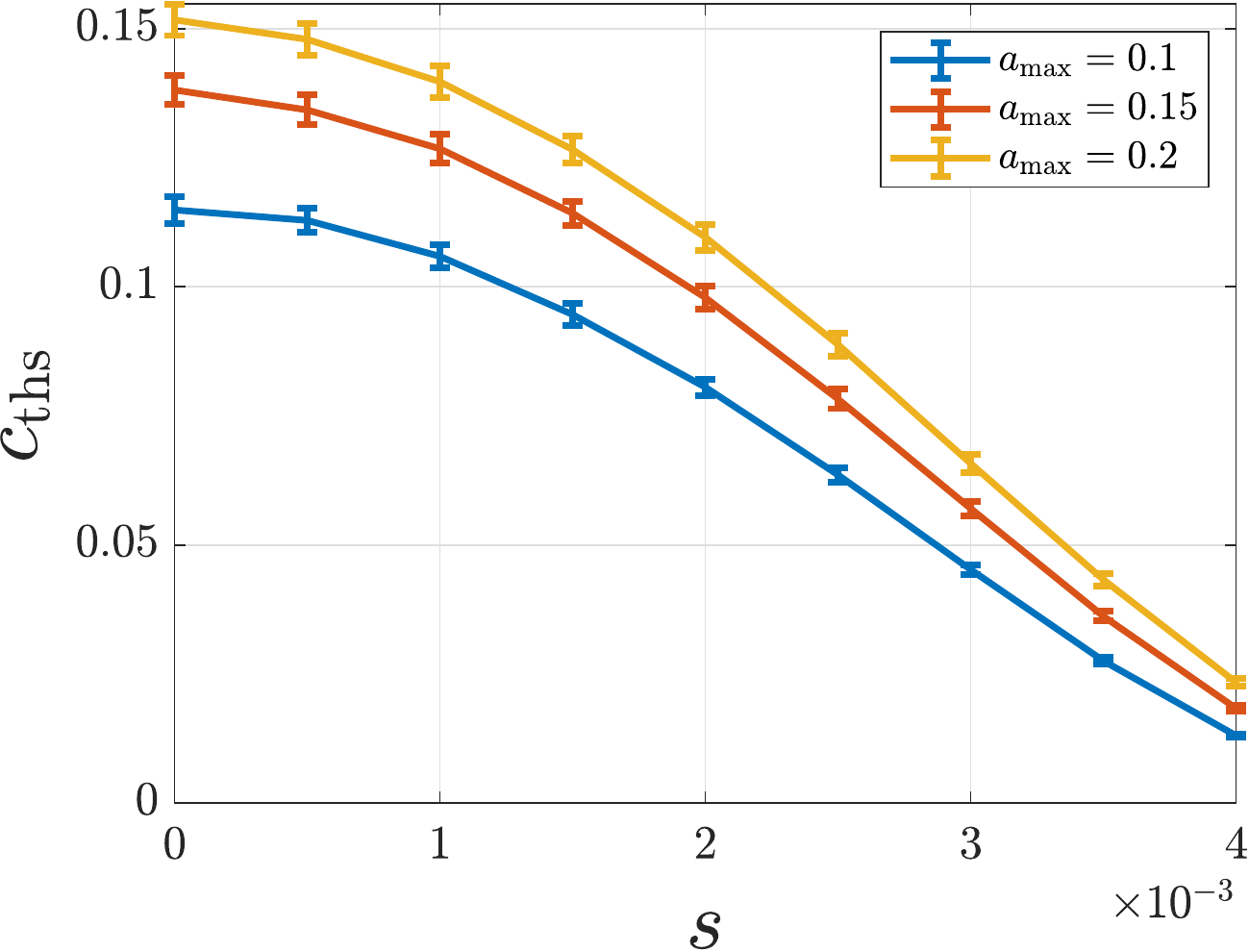}}\\
    \subfloat[Porosity usage]{\label{fig:usage_flux_threshold}\includegraphics[width=.45\textwidth,height=.4\textwidth]{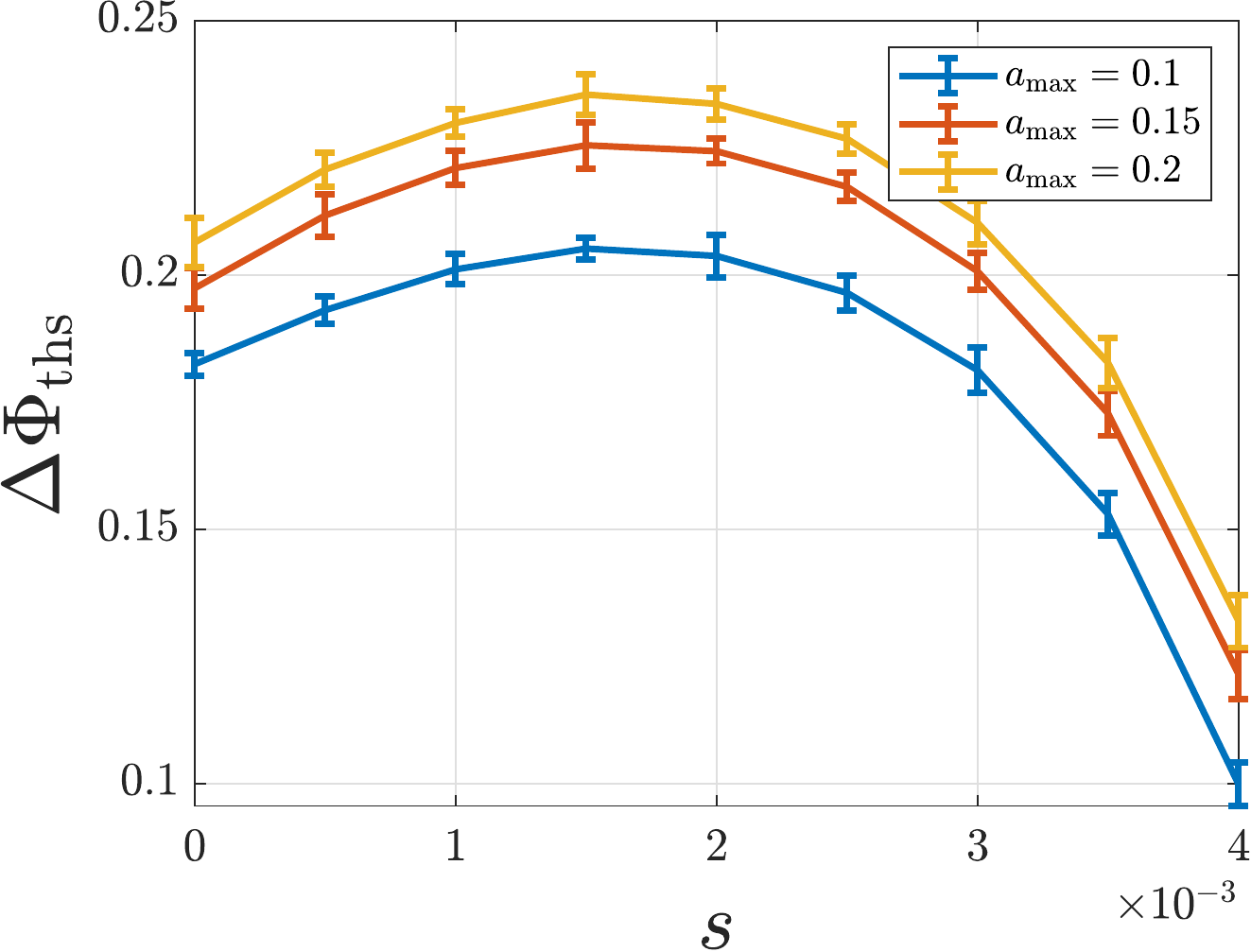}}
    \caption{Performance metrics under flux threshold versus radius gradient. Flux Threshold is set at $2\times 10^{-6}$.}
    \label{fig:flux_threshold}
\end{figure}

\cref{fig:flux_threshold} shows the performance metrics of radius-graded membrane networks where filtration is halted after the flux level drops below $2\times 10^{-6}$. This threshold level is approximately $30\%$ of the initial flux for uniform networks, and roughly $80\%$ of that for the steepest-graded network with gradient value $s=4\times 10^{-3}$ (see the vertical scale of \cref{fig:q0}). From this observation on the flux threshold level alone, we anticipate that filters with smaller initial fluxes, namely, networks with large pore-size gradients, are more prone to halt filtration prematurely and are thereby disadvantageous under this filtration mode. 

\cref{fig:h_flux_threshold} shows total filtrate throughput against radius gradient $s$ for pore-size-graded networks that operate until they reach the imposed flux level. With the chosen parameters and stopping criterion, we again observe a maximizing gradient value at $s=1.5\times 10^{-3}$, a value smaller than that in \cref{fig:tt} ($s=2\times 10^{-3}$), where networks operate until flux extinction. Thus, under the new threshold-based stopping criterion, filters with smaller gradient (and hence larger initial flux) are more favored in terms of throughput production. In fact, networks with $s\geq 3\times 10^{-3}$ underperform quite significantly, even relative to uniform networks, because their total filtering time is greatly shortened under the imposed flux threshold. \cref{fig:c_flux_threshold} shows final accumulated concentration of foulant (measured at the membrane outlet) against $s$. Here, we observe a monotone trend in both $s$ and in $a_{\rm max}$. These trends maintain qualitatively the same features as in \cref{fig:c}, though we observe that here the concentration is pointwise (for every $s$) larger than that in \cref{fig:c}. This is expected because filtration is stopped prematurely (in our simple model the particle capture capability of the membrane improves continuously as fouling occurs and pores shrink). In \cref{fig:usage_flux_threshold}, we show the relationship between the membrane porosity usage $\Delta \Phi_{\rm ths}$ and radius gradient $s$ under the imposed flux threshold. We observe a clear maximum in $\Delta \Phi_{\rm ths}$ at $s=1.5\times 10^{-3}$. 

The results in \cref{fig:flux_threshold} imply that with the imposed lower threshold on fluid flux, membrane networks with a radius gradient of $s=1.5\times 10^{-3}$ should be preferred over others due to their combined score of filtrate production, particle retention capabilities, and porosity usage. Once again, we note that this optimal value is almost independent of the value of $a_{\rm max}$ considered. We would, however, anticipate that the optimal radius gradient will decrease if the imposed lower flux threshold is increased (and it would, of course, change if model parameters were changed). 

\section{Conclusion}\label{sec:conclusion_grad}

In conclusion, we have devised a general procedure to generate pore-size-graded banded membrane pore networks, representing multilayered membrane filters. We have studied the influence of the pore-size (radius) gradient $s$, and maximum pore length $a_{\rm max}$, on selected performance metrics of these networks, under two setups of relevance to applications -- filtration until flux extinction, or until a flux lower threshold is reached. For the chosen model parameters, we have also determined optimizing pore-size gradient values for some of the performance metrics considered (compiled in \cref{table:optimal_values}).

When filters run to extinction, we find that total filtrate throughput exhibits a non-monotone trend against pore-radius gradient. More precisely, for the parameters we considered, membrane networks with a pore-radius gradient value of $s=2\times 10^{-3}$ achieve maximal total filtrate throughput. However, accumulated foulant concentration at the membrane outlet is monotonically decreasing in $s$, suggesting that, for foulant control purposes only, one should prefer membrane networks with a pore-radius gradient as large as possible. To examine the extent of membrane fouling, we also study the porosity change of the entire membrane over its lifetime (per \cref{eq:delta_membrane_porosity}). This quantity is found to be non-monotone in $s$, with a pronounced maximum achieved at a pore-radius-gradient value of $s=3\times 10^{-3}$ (for the chosen parameter values). To determine the extent of foulant penetration, we also investigate the change in porosity of each band (per \cref{eq:delta_band_porosity}). We find that the porosity change in the first band is monotone in $s$, while that in other bands has a clear (but different) maximizing pore-size gradient value. In particular, the gradient values that maximize porosity changes for downstream ($3^{\rm rd}$ and $4^{\rm th}$) bands are very close to the gradient value that maximizes total throughput, suggesting a strong correlation between these performance measures.

However, when we stop the filtration at a prescribed minimum flux level, we observe that the optimal pore-radius gradient for each performance metric is smaller than when filters are run until flux extinction. For the chosen model parameters, total filtrate throughput and porosity usage are all maximized at a gradient value of $s=1.5\times 10^{-3}$ under the flux threshold criterion (final accumulated foulant concentration at the membrane pore outlets remains a monotone decreasing function in gradient $s$). The fact that we observe a smaller optimal gradient here than with the flux-exhaustion stopping criterion is mainly because of the advantage given by the flux threshold to filters with large initial flux. Uniform networks ``benefit" from this practice and rise up the ranks into the better performing filters. At the same time, graded networks with large pore-radius gradients perform poorly because filtration tends to halt at an early stage due to the small initial fluxes inflicted by the high-resistance downstream pores. We also anticipate that the optimal gradient value(s) for performance metrics considered in this work will depend on the flux threshold we impose (indeed, on the basis of \cref{fig:q0} we expect that $s=0$ may become the optimal value when the imposed flux threshold is high enough).

\begin{table}
\caption{\label{table:optimal_values}Optimal Radius Gradient Value for Each Performance Metric.}
\begin{ruledtabular}
\begin{tabular}{p{0.45\linewidth} | p{0.18\linewidth}|p{0.3\linewidth} }
Performance Metric (\cref{sec:perf_metrics}) & Metric Symbol & Optimal Radius Gradient \\
\midrule
{\bf Until flux extinction} & & \\
\hline
Total throughput & $h_{\rm final}$ & $2\times 10^{-3}$  \\
Initial flux & $q_{\rm out}\left(0\right)$ & $0$  \\
Accumulated concentration of foulant at membrane outlet & $c_{\rm final}$ & $0$  \\
Membrane porosity usage & $\Delta \Phi$ & $3\times 10^{-3}$  \\
\midrule
{\bf Until flux threshold} & & \\
\hline
Total throughput  & $h_{\rm ths}$ & $1.5\times 10^{-3}$  \\
Accumulated concentration of foulant at membrane outlet & $c_{\rm ths}$ & $0$  \\
Membrane porosity usage & $ \Delta \Phi_{\rm ths}$ & $1.5\times 10^{-3}$  \\
\end{tabular}
\end{ruledtabular}
\end{table}

We also found that the observed trends in pore-radius gradient persist for all values of maximal pore length $a_{\rm max}$ considered. This suggests that our findings of how performance metrics depend on pore-radius gradient are largely independent of variations in membrane interior microstructure as characterized by $a_{\rm max}$. 

We note that we did not vary $\lambda$ (a band-independent parameter that captures particle-membrane affinity) in this study. However, the effect of variations in $\lambda$ simply varies the time scale of the problem, which we anticipate will not affect the overall trends of performance metrics against pore-radius gradient that were observed with the chosen value of $\lambda$. One idea for future work would be to introduce band-specific $\lambda_k$'s, which represent multilayered membrane filters consisting of different materials.

Future work should also include intra-layer pore-size variations in pore-radius graded banded networks. A preliminary analytical result (per~\cref{app:analytical}) based on the governing equations derived from our model of this paper suggests that membrane pore networks with constant band radius always close upstream; that is, the membrane will never stop functioning due to critical pore closures in the interior of the network but only when the radii of all inlets on the top surface are zero. With intra-layer pore-size variations, adsorptive behaviors at the global scale may become more complicated and more interesting than the constant band radius case. Additionally, other fouling mechanisms such as sieving could be modeled to provide a more complete picture of the membrane filtration process.

\begin{acknowledgments}
This work was supported by NSF Grants No. DMS-1615719, DMS-2133255 and DMS-2201627.
\end{acknowledgments}

\appendix
\section{Junctions and Pores in a Band}\label{app:vertex_edge_defn}
In this appendix, we define the set of junctions (vertices) and pores (edges), and their respective band-specific counterparts. We work exclusively with the dimensionless variables defined in \cref{sec:grad_scales}. Our work treats junctions and pores as points and straight lines respectively, which lie in our dimensionless domain -- the unit cube (though the notions of vertices and edges are generally more abstract in classical graph theory). Each junction $v$ of the junction set $\mathcal{V}$ has a Euclidean coordinate $v = \left(v_x,v_y,v_z\right) \in \left[0,1\right]^3$, with $z$ measured from the membrane top ($z=0$) to bottom surface ($z=1$), with $z=(Z-0.5)/W$. The junction coordinates are generated randomly as described in \cref{sec:band_and_membrane_porosity}. We further define the set of membrane pore inlets and outlets, 
\begin{subequations}
    \begin{align}
        \mathcal{V}_{{\rm in}} & =\left\{ v\in V:v_{z}=0\right\}, \\
        \mathcal{V}_{{\rm out}} & =\left\{ v\in V:v_{z}=1\right\}. 
    \end{align}
    \label{eq:grad_vtop_vbot}
\end{subequations}

The set of edges $\mathcal{E}$ is formed by connecting the junctions via
\begin{equation}
    \mathcal{E}=\left\{ e_{vw}\in \mathcal{V}\times \mathcal{V}:a_{\rm min}<\chi\left(v,w\right)<a_{\rm max}\right\}, \,\, a_{\rm max} \leq \frac{1}{m},
    \label{eq:edge}
\end{equation}
where $a_{\rm min}$ and $a_{\rm max}$ are the dimensionless minimum and maximum distance allowed between two junctions respectively; $\chi\left(\cdot,\cdot\right)$ is a periodic metric, defined by
\begin{equation}
    \chi\left(v,w\right) = \min_{\left(w_x,w_y\right)}\left\Vert v-\left(w_x,w_y,0\right)\mid w_x,w_y\in \left\{\pm 1, 0\right\}\right\Vert _{2},
    \label{eq:connection_metric}
\end{equation}
that is, junctions close to the four sides parallel to the $z$-direction may be connected through the boundary. We constrain $a_{\rm max}$ so that it does not exceed the thickness of a band, otherwise edges may cross more than two bands and reduce or defeat the purpose of having a gradient of pore radii. 

Next, we define precisely junctions and edges within a given band. Denoting the $k^{\rm th}$ band as the set of coordinates
\begin{equation}
    \mathcal{V}_{k}=\left\{ \left. v\in\left[0,1\right]^3\right|\frac{k-1}{m}\leq v_{z}<\frac{k}{m}\right\},
    \label{eq:banded_vertex_set}
\end{equation}
we say a junction $w$ lies in the $k^{\rm th}$ band if $w\in \mathcal{V}_k$. We treat each edge as a straight line in the unit cube, 
\[
e_{vw}=\left\{ u\in\left[0,1\right]^{3}\mid u=\zeta v+\left(1-\zeta\right)w,\quad0\leq \zeta \leq1\right\}.
\]
Let $L\left(e_{vw}\right)$ be the one-dimensional Lebesgue measure of $e_{vw}$ such that $L\left(e_{vw}\right) = \chi\left(v,w\right)$. Define a band-specific length measure $L_k$ such that
\begin{equation}
    L_{k}\left(e_{vw}\right) = L\left(e_{vw}\cap \mathcal{V}_k\right),
    \label{eq:intersection_measure}
\end{equation}
which computes the length of the edge strictly inside the $k^{\rm th}$ band (known in general as an intersection measure). We say that a pore belongs to the $k^{\rm th}$ band when the largest proportion of its length lies strictly inside the $k^{\rm th}$ band. More precisely, we define the set of the pores in the $k^{\rm th}$ band as
\begin{equation}
   \mathcal{E}_k = \left\{e_{vw}\in \mathcal{E} \mid L_{k}\left(e_{vw}\right) = \max_{n} L_{n}\left(e_{vw}\right)\right\}.
   \label{eq:edge_set}
\end{equation}
In this definition, we see that if $v,w\in \mathcal{V}_k$, then $e_{vw}\cap \mathcal{V}_k = e_{vw}$ while $e_{vw}\cap \mathcal{V}_n = \emptyset$ for all $n\neq k$, that is, the edge $e_{vw}$ lies strictly in the $k^{\rm th}$ band.

The formula \cref{eq:intersection_measure} also facilitates the computation of $\Phi_k$, the band porosity of the $k^{\rm th}$ band (\cref{eq:band_porosity}), in the sense that we consider the lengths of edges that strictly lie in $\mathcal{V}_k$; edges reaching two bands will contribute to the band porosities of each band separately. We simplify the notation as
\[
L_{k,vw}:=L_{k}\left(e_{vw}\right).
\]

\section{Number of Random Points in Each Band}\label{app:band_estimate}
With prescribed $\Phi$ and $m$, we provide an estimate of how many random points, $N_k$, should be used in the $k^{\rm th}$ band. We write total pore length as $\chi_{ij}:=\chi\left(i,j\right)$ per \cref{eq:connection_metric}.  More precisely, we use basic arguments to deduce, via the sequence of approximations and equalities below, that total edge length in the $k^{\rm th}$ band scales with $N_k^2$,
\begin{equation}
\sum_{e_{ij}\in \mathcal{E}} l_{k,ij} \overset{\left(\textbf{A}\right)}{\approx} \sum_{e_{ij}\in \mathcal{E}_k} \chi_{ij} \overset{\left(\textbf{B}\right)}{=}  \overline{\chi} \left|\mathcal{E}_k\right| \overset{\left(\textbf{C}\right)}{=} \overline{\chi} \frac{\overline{d}_k N_k}{2}
\overset{\left(\textbf{D}\right)}{=} \frac{\overline{\chi}}{2} \left[\left(N_k-1\right) p\right] N_k,
\label{eq:approximation}
\end{equation}
where $\overline{\chi}$, $\overline{D}_k$ and $p$ are the average edge length, average number of neighbors in the $k^{\rm th}$ band and the probability of two random points being connected, respectively.

The first approximation $\left(\textbf{A}\right)$ relies on the estimate that $l_{k,ij}\approx \chi_{ij}$, given $e_{ij}\in \mathcal{E}_k$. The first equality $\left(\textbf{B}\right)$ is trivial. The second equality $\left(\textbf{C}\right)$ expresses the number of edges in the $k^{\rm th}$ band, $\left|\mathcal{E}_k\right|$, as the average number of neighbors, $\overline{d}_k$, times the number of junctions, $N_k$, divided by 2 (to account for double counting). The third equality $\left(\textbf{D}\right)$ expresses $\overline{d}_k$, as the total number of neighbors a junction could have, $N_k-1$, times the probability of obtaining a neighbor, $p$. 

We now briefly discuss the terms to the right of the last equality $\left(\textbf{D}\right)$ in \cref{eq:approximation} and the dependence of these quantities on band number $k$. First, since $\overline{\chi}$ is a sample mean of edge length with $\left|\mathcal{E}_k\right|$ as the sample size, it can be approximated by the expected edge length (based on hypercube line-picking~\cite{hypercube}), which is a constant independent of $k$ (but dependent on $a_{\min}$ and $a_{\max}$). The probability of connecting two uniformly random points in the cube, $p$, depends on $a_{\min}$ and $a_{\max}$ but not on band number $k$.  These simple estimates provide the basis that justifies the step from \cref{eq:local_porosity_approx} to \cref{eq:num_k_relationship_1}, that is, one can cancel $\overline{\chi}$ and $p$ from both sides of \cref{eq:local_porosity_approx} after re-expressing $\sum_{e_{ij}\in \mathcal{E}} l_{k,ij}$ using \cref{eq:approximation}.

We note that $\overline{\chi}$ does depend on $k$ since each band is expected to have different node density. Nonetheless, using the expected value of average edge length as an approximation is a reasonable starting point to help estimate the number of nodes needed in each band.  The procedure of edge addition/removal performed in \cref{al:correction} of \cref{al:algorithm1} is considerably sped up with the guided initial guesses.

\section{Analytical Results on Pore Closure Time}\label{app:analytical}

\begin{figure}
    \centering
    \includegraphics[scale=0.4]{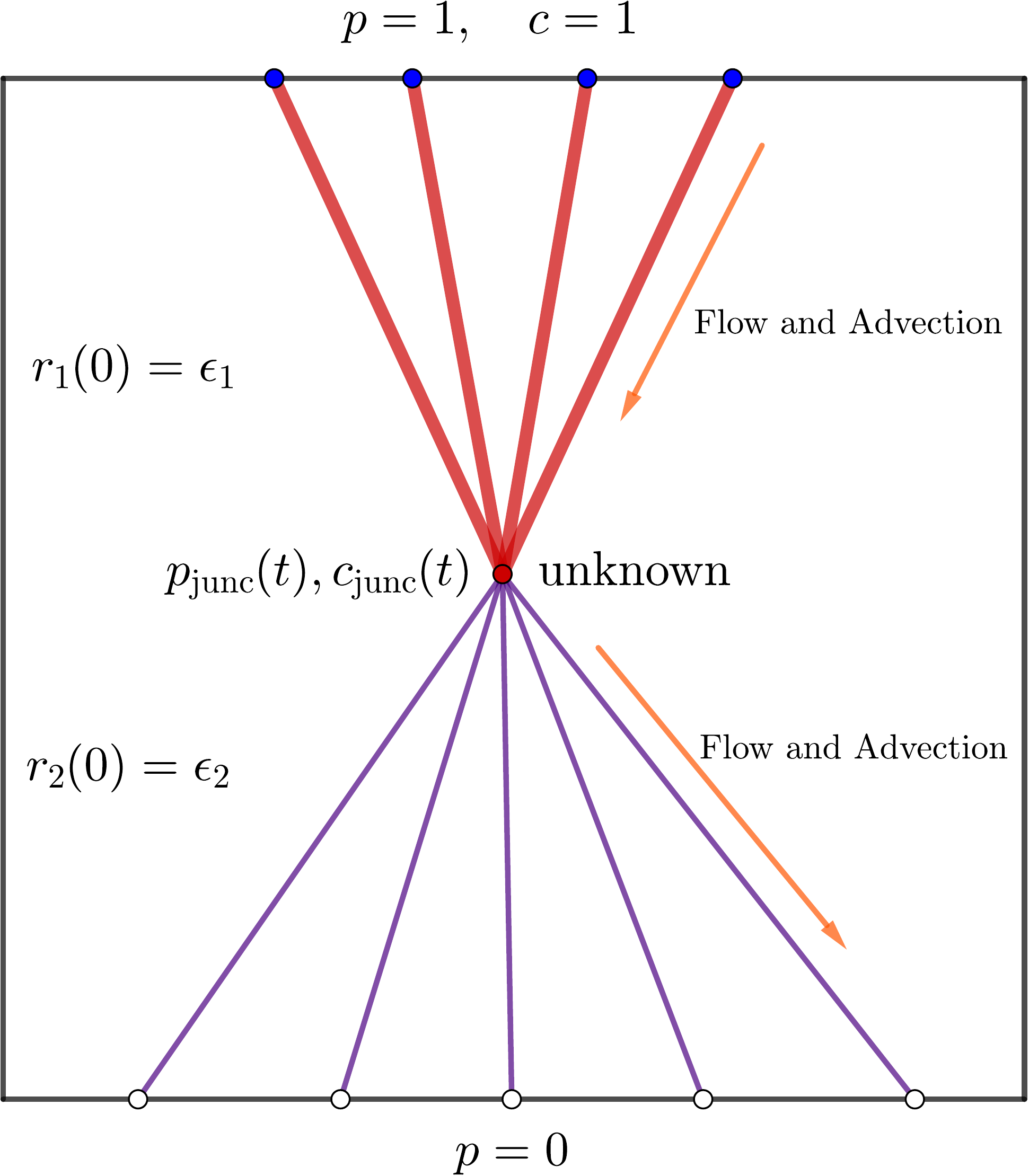}
    \caption{Schematic of a simplified setup for \cref{thm:thm1}. Colored junctions and pores correspond to each band as follows: {\color{red}red upstream pores} and {\color{indigo}indigo downstream pores}. {\color{blue}Blue dots} are inlets. {\color{red}The red dot} is an interior junction. White dots are outlets. Here $n_{\rm up} = 4$ and $n_{\rm down}=5$.
    }
    \label{eq:analytical_schematic}
\end{figure}

In this section, we show that the lifetime of a simple subclass of pore networks is governed by the radius of the inlets (pores in the upstream surface). More precisely, for such simple networks, the radii of the inlets will always go to zero earlier than those of downstream pores, {\bf independent} of the initial upstream and downstream pore radii and model parameters. This result forms the basis for the conjecture that the more general networks considered in this work will also only clog at the membrane inlet. It also serves as an instructive worked example for the general network solver (see Gu {\it et al.}~\cite{gu_network_2021} for another such example). 

A membrane pore network ceases to function as a filter when there no longer exists a feasible path connecting any inlet on the top surface to any outlet on the bottom surface. The critical event leading to filtration arrest is when the radius of a pore vanishes as the ``last straw", breaking the main network into at least two disconnected subnetworks, such that each subnetwork contains only a subset of the inlets or outlets, but not both. We here consider the subclass of networks (depicted in \cref{eq:analytical_schematic}) consisting of a single arbitrary pore junction connecting $n_{\rm up}$ upstream inlets and $n_{\rm down}$ downstream outlets. 

The upstream and downstream pores are assumed all to have unit length (though they appear to have different lengths in \cref{eq:analytical_schematic}; the choice of a common length simplifies the presentation but does not affect our result). We solve the dimensionless governing equations (per scales presented in \cref{sec:grad_scales}) for the unknown pressure $p_{\text{junc}}\left(t\right)$ and concentration $c_{\text{junc}}\left(t\right)$ at the interior junction. Since the pore lengths are assumed the same, and the upstream (resp. downstream) pores obey the same boundary conditions for pressure and concentration, the radius and concentration evolution in these pores are therefore also the same. As a result, we simply monitor the evolution of quantities for one upstream and one downstream pore. Let $r_1\left(t\right)$ and $r_2\left(t\right)$ be the radius of each upstream and downstream pore respectively. 

Fluxes through each upstream and downstream pore, labelled $q_1\left(t\right)$ and $q_2\left(t\right)$ respectively, satisfy the dimensionless Hagen-Poiseuille equations according to the pressure drop across them,
\begin{align*}
q_{1}\left(t\right) & =\left(1-p_{\text{junc}}\left(t\right)\right)r_{1}^{4}\left(t\right),\\
q_{2}\left(t\right) & =p_{\text{junc}}\left(t\right)r_{2}^{4}\left(t\right),
\end{align*}
where $p_{\text{junc}}\left(t\right)$ is the (unknown) pressure at the interior junction. Conservation of flux yields
\begin{equation}
    n_{{\rm up}}q_{1}\left(t\right)=n_{{\rm down}}q_{2}\left(t\right)
    \label{eq:analytical_conservation_of_flux}
\end{equation}
and therefore 
\begin{equation}
p_{\text{junc}}\left(t\right)=\frac{n_{{\rm up}}r_{1}^{4}\left(t\right)}{n_{{\rm up}}r_{1}^{4}\left(t\right)+n_{{\rm down}}r_{2}^{4}\left(t\right)} \qquad \implies \qquad q_{1}\left(t\right)=\frac{n_{{\rm down}}r_{1}^{4}\left(t\right)r_{2}^{4}\left(t\right)}{n_{{\rm up}}r_{1}^{4}\left(t\right)+n_{{\rm down}}r_{2}^{4}\left(t\right)}.
\label{eq:analytical_flux}
\end{equation}

Foulant concentration in the upstream pore, $c_1\left(y,t\right)$, satisfies the dimensionless advection equation \cref{eq:advection},
\[
q_{1}\frac{\partial c_{1}}{\partial y}=-\lambda r_{1}c_{1},\quad c_{1}\left(0,t\right)=1,
\]
which has an analytical solution
\[
c_{1}\left(y,t\right)=c_1\left(0,t\right)\exp\left(-\frac{\lambda yr_{1}\left(t\right)}{q_{1}\left(t\right)}\right)\overset{\cref{eq:analytical_flux}}{=}\exp\left(-\lambda y\left[\frac{n_{{\rm up}}r_{1}\left(t\right)}{n_{{\rm down}}r_{2}^{4}\left(t\right)}+\frac{1}{r_{1}^{3}\left(t\right)}\right]\right).
\]
The evolution of the pore radii satisfies \cref{eq:adsorption},
\begin{align*}
\frac{dr_{1}}{dt} & =-1,\quad r_{1}\left(0\right)=\epsilon_{1} \qquad \implies \qquad  r_{1}\left(t\right)=\epsilon_{1}-t,\\
\frac{dr_{2}}{dt} & =-c_{\text{junc}}\left(t\right),\quad r_{2}\left(0\right)=\epsilon_{2}.
\end{align*}
By conservation of particle flux at the junction, we have $n_{{\rm down}}q_{2}\left(t\right)c_{\text{junc}}\left(t\right)=n_{{\rm up}}q_{1}\left(t\right)c_{1}\left(1,t\right)$. Conservation of flux (per~\cref{eq:analytical_conservation_of_flux}) reduces this to $c_{\text{junc}}\left(t\right)=c_{1}\left(1,t\right)$. Hence,
\begin{equation}
    \frac{dr_{2}}{dt} =-c_{1}\left(1,t\right)=-\exp\left(-\lambda\left(\frac{n_{{\rm up}}\left(\epsilon_{1}-t\right)}{n_{{\rm down}}r_{2}^{4}\left(t\right)}+\frac{1}{\left(\epsilon_{1}-t\right)^{3}}\right)\right),\quad r_2\left(0\right) = \epsilon_2.
    \label{eq:r2}
\end{equation}

\begin{theorem}\label{thm:thm1}
The solution $r_2\left(t\right)$ to \cref{eq:r2} satisfies $r_{2}\left(t\right)>0$ for all $t\in\mathcal{T}:=\left[0,\epsilon_{1}\right]$, for all $\epsilon_1,\epsilon_2,\lambda>0$ and arbitrary positive integers $n_{\rm up}$ and $n_{\rm down}$.
\end{theorem}

\begin{proof}
We note first that $r_{2}\left(t\right)\geq0$ for all $t\in\mathcal{T}$ since
the initial condition is positive ($\epsilon_{2}>0$), and the right hand side of \cref{eq:r2} is a nonpositive function, which goes to zero as $r_2\rightarrow 0$, {\it i.e.}, $\frac{dr_{2}}{dt}\rightarrow 0^{-}$ as $r_{2}\rightarrow0^{+}$. Thus, the radius of the downstream pore will decrease to zero until it reaches zero and will never attain a negative value. To prove the claim that $r_{2}\left(t\right)>0$ for all $t\in\mathcal{T}$, we suppose that there exists $t^* \in \mathcal{T}$ such that $r_2\left(t^*\right) = 0$ and arrive at a contradiction as follows.

The above shows that $r_2\left(t\right)$ is a monotone decreasing function, and in fact, is equal to $0$ for all $t\geq t^*$. While $r_2\left(t\right)>0$, for $0\leq t<t^*$, we divide both sides of \cref{eq:r2} by $r_2\left(t\right)$ and integrate to obtain 
\begin{equation}
    \log \left(r_2\left(t\right)\right) = \log\left(\epsilon_2\right) - I\left(t\right),
    \label{eq:contradiction_step}
\end{equation}
where 
\begin{equation}
    I\left(t\right) = \int_{0}^{t}\frac{1}{r_{2}\left(\tau\right)}\exp\left(-\lambda\left[\frac{n_{{\rm up}}\left(\epsilon_{1}-\tau\right)}{n_{{\rm down}}r_{2}^{4}\left(\tau\right)}+\frac{1}{\left(\epsilon_{1}-\tau\right)^{3}}\right]\right)d\tau.
    \label{eq:integrand}
\end{equation}
Note that under the assumption that $r_2\left(t^*\right) = 0$, \cref{eq:contradiction_step} requires $I\left(t\right) \rightarrow +\infty$ as $t\nearrow t^*$. It therefore suffices to check that the integrand defining $I(t)$ in \cref{eq:integrand} is a bounded function 
for all $t \in \mathcal{T}$ to obtain our contradiction. 

The term $\frac{1}{\left(\epsilon_{1}-\tau\right)^{3}}$ is unimportant in the exponential of \cref{eq:integrand} since the integral without it bounds $I\left(t\right)$ from above. We focus on checking the boundedness of the following part
\[
f\left(\tau\right) := \frac{1}{r_{2}\left(\tau\right)}\exp\left(-\kappa\frac{\left(\epsilon_{1}-\tau\right)}{r_{2}^{4}\left(\tau\right)}\right),
\]
where $\kappa = \frac{\lambda n_{{\rm up}}}{n_{{\rm down}}}$.

Away from $r_2 = 0$, the integrand is clearly bounded. As $r_2\rightarrow 0$ (or as $\tau \rightarrow t^*)$, by L'H{\^ o}pital's rule, $f\left(\tau\right) \rightarrow 0$. The integrand is thus void of singularities and does not blow up on $[0,t]$. This implies that the RHS of \cref{eq:contradiction_step} is bounded for all $t\in \mathcal{T}$, while the LHS gives $-\infty$ as $t\rightarrow t^*$. This contradiction shows that there is no $t^* \in \mathcal{T}$ such that $r_2(t^*) = 0$.
\end{proof}

\cref{thm:thm1} shows that when initial upstream (resp. downstream) pore sizes are the same, pore closure always occurs upstream. This result is useful for the majority of the pore junctions considered in this work, {\it i.e.}, those that connect upstream and downstream pores with the same initial radius respectively. 

However, \cref{thm:thm1} did not consider the possibility for a junction to have downstream pores with different initial radii, as may occur when a junction connects to downstream pores that belong to different bands. In this case, we show via a numerical example that the radius of a downstream pore can go to zero before the upstream ones. Consider an inverted-Y shaped network with one upstream pore and two downstream pores (corresponding to $n_{\rm top} = 1$ and $n_{\rm bot} = 2$ in \cref{eq:analytical_schematic}), with initial radii $\epsilon_1$, $\epsilon_2$ and $\epsilon_3$ respectively. Without loss of generality, we assume $\epsilon_2>\epsilon_3$. Note that this situation is not included in the premise of \cref{thm:thm1}.

Using calculations similar to those used for \cref{thm:thm1} above, we find that the outlet radii $r_2\left(t\right)$ and $r_3\left(t\right)$ satisfy the set of coupled ODEs, 
\begin{equation}
    \frac{dr_2}{dt} = \frac{dr_3}{dt}= -\exp\left(-\lambda\left(\frac{\epsilon_{1}-t}{r_{2}^{4}\left(t\right)+r_{3}^{4}\left(t\right)}+\frac{1}{\left(\epsilon_{1}-t\right)^{3}}\right)\right), \quad r_2\left(0\right)=\epsilon_2, \, r_3\left(0\right)=\epsilon_3.
\end{equation}
Note that $r_2$ and $r_3$ simply differ by a constant, $\epsilon_2 - \epsilon_3$. Thus, we can further deduce that
\begin{equation}
    \frac{dr_3}{dt}= -\exp\left(-\lambda\left(\frac{\epsilon_{1}-t}{\left(r_3\left(t\right)+\epsilon_2-\epsilon_3\right)^4+r_{3}^{4}\left(t\right)}+\frac{1}{\left(\epsilon_{1}-t\right)^{3}}\right)\right).
\end{equation}
From this, we observe that one can make $\epsilon_2-\epsilon_3$ sufficiently large so that regardless of how small $r_3$ becomes, $r_3$ decreases at a nontrivial rate. This is a scenario different than that in \cref{eq:r2}. An explicit condition  involving $\epsilon_1,\epsilon_2,\epsilon_3$ and $\lambda$ for this to happen may be derived. We have confirmed this conclusion numerically with $\epsilon_1= \epsilon_2 = 0.01$, $\epsilon_3=3\times 10^{-3}$, $\lambda = 5\times 10^{-7}$ (the same value used in \cref{sec:grad_results}), finding that $r_3\left(t\right)$ goes to zero earlier than $r_1\left(t\right)$. 

We conclude that the difference in initial conditions for pore radii does play a role in driving the dynamics of each pore. Downstream pore closure can be earlier than the upstream one. However, we stress that even if one downstream pore closes earlier, the local structure at the junction always reduces to the case where we have multiple upstream pores and one single downstream one, which is the setup used in \cref{thm:thm1} with $n_{\rm up}$ arbitrary and $n_{\rm down} = 1$. In other words, the moment any junction has one downstream pore, its upstream pores will always close first. With this heuristic argument (that can be argued inductively upstream), we believe that a general membrane network with varied initial conditions on the pore radii will only close on the top surface under adsorptive fouling. 

\nocite{*}

\bibliography{main.bib}

\end{document}